\documentclass{article}

\usepackage[margin=1in]{geometry}

\usepackage{amsmath, amssymb, amsfonts}

\usepackage{amsthm}
\theoremstyle{plain}
\newtheorem{theorem}{Theorem}[section]
\newtheorem{lemma}[theorem]{Lemma}

\theoremstyle{remark}
\newtheorem{remark}[theorem]{Remark}

\theoremstyle{definition}
\newtheorem{example}[theorem]{Example}

\usepackage{newtxtext}
\usepackage[subscriptcorrection]{newtxmath}

\usepackage{graphicx}
\graphicspath{{figures/}{}}

\usepackage[plain,noend]{algorithm2e}
\makeatletter
\renewcommand{\algocf@captiontext}[2]{#1\algocf@typo. \AlCapFnt{}#2}

\def\@algocf@capt@plain{top}
\renewcommand{\algocf@makecaption}[2]{%
  \addtolength{\hsize}{\algomargin}%
  \sbox\@tempboxa{\algocf@captiontext{#1}{#2}}%
  \ifdim\wd\@tempboxa >\hsize%
    \hskip .5\algomargin%
    \parbox[t]{\hsize}{\algocf@captiontext{#1}{#2}}%
  \else%
    \global\@minipagefalse%
    \hbox to\hsize{\box\@tempboxa}%
  \fi%
  \addtolength{\hsize}{-\algomargin}%
}
\makeatother

\usepackage{natbib}

\usepackage{booktabs}
\usepackage{subcaption}
\usepackage{float}

\usepackage{hyperref}

\protected\def\[#1\]{\begin{equation}\begin{aligned}#1\end{aligned}\end{equation}}
\protected\def\(#1\){\begin{equation*}\begin{aligned}#1\end{aligned}\end{equation*}}

\newenvironment{keywords}{\smallskip\noindent\textit{Keywords:}}{\par}

\title{Relaxation of Projected Prior with Continuous Gap Shrinkage}

\author{
  Leo L. Duan\thanks{Department of Statistics, University of Florida, Gainesville, Florida 32611, U.S.A. Email: \texttt{li.duan@ufl.edu}.}
  \and
  Sunghyun Cho\thanks{Department of Statistics, University of Florida, Gainesville, Florida 32611, U.S.A. Email: \texttt{sunghyuncho@ufl.edu}.}
  \and
  Mingzhang Yin\thanks{Department of Marketing, University of Florida, Gainesville, Florida 32611, U.S.A. Email: \texttt{mingzhang.yin@warrington.ufl.edu}.}
}

\date{}

\begin{document}

\maketitle

\begin{abstract}
Projected priors were originally introduced to accommodate parameter constraints, but have recently regained popularity due to their ability to assign probability mass to low-dimensional parameter sets, such as the spaces of sparse vectors, directed acyclic graphs, or transport plans. When employed as a transformation of random variables, projection is especially useful, since its contraction property not only preserves probability concentration, but also often preserves differentiability for gradient-based posterior computation.
On the other hand, unless the projection can be obtained by some non-iterative algorithm, posterior computation can be expensive because it requires nesting an iterative optimization routine within each Markov chain Monte Carlo iteration. In this article, inspired by the success of continuous shrinkage models as replacements for discrete spike-and-slab priors, we propose a continuous relaxation of projected priors. The key idea is to quantify the duality gap between the primal projection loss and the dual objective, and impose a probabilistic prior that shrinks this gap toward zero. The resulting gap-shrinkage prior has a tractable form, does not require running an optimization subroutine inside each posterior update, and puts probability mass near the exact projection. We demonstrate useful properties of gap-shrinkage priors, including connections to global-local shrinkage priors, broad applicability to generalized projection functions, and competitive performance in posterior contraction. We apply the gap-shrinkage model to a marketing data analysis aimed at identifying important predictor effects on multivariate grocery-shopping decisions.
\end{abstract}


\section{Introduction}
Modern statistical applications often face problems of structured inference: the parameter of interest is not free to vary over a full Euclidean space, but is instead constrained to lie in a set with scientific or geometric meaning. Such a structure arises, for example, when regression functions are shape-restricted, coefficient vectors are sparse, matrices are low rank, or graphs obey combinatorial restrictions. In the Bayesian literature, there are at least three broad ways to encode this information. A classical approach is to place a prior directly on the constrained space, as in constrained-support Bayesian analysis \citep{gelfand1992bayesian,dunson2003order,neelon2004isotonic,mohammadi2015bayesian,li2019graphicalhorseshoe,agrell2019gaussian}. A second approach is to relax the constraint by defining prior mass through distance to the target set, thereby replacing a hard support restriction with a continuous neighborhood around the constrained space \citep{duan2020constraint,presman2023distance}. A third approach is to begin with a continuous ambient random variable and map it to the structured set by projection \citep{lin2014bayesian,astfalck2018posterior,xu2023l1ball,xu2024proximal,thompson2024prodag,bhaumik2022two,pal2025projection}. These three strategies all aim to respect structure, but they differ sharply in how they distribute prior mass and in the computational burdens they impose.

Projection differs from the other two approaches in a particularly important way: it does not only confine probability mass in the constrained set, but transports positive mass onto or near the lower-dimensional faces of that set, such as its boundary. This feature is often statistically desirable because the boundary frequently contains the most interpretable or parsimonious objects. For the ball generated by the $\ell_1$-norm, for example, boundary points include vectors with exact zeros, so projection naturally yields sparse representations \citep{duchi2008efficient}. For the nuclear-norm ball, even though the interior contains matrices of all ranks \citep{fazel2001rankminimization}, the boundary contains matrices with some singular values equal to zero, hence suitable for low-rank matrix estimation. For the space of joint probability matrix (one with non-negative entries that sum to one), the boundary contains matrices with some entries equal to zero, hence suitable for estimating a parsimonious transport plan between two marginal distributions \citep{peyre2019computational}. This perspective is reflected in the broader proximal-mapping Bayesian approach that harnesses generalized projection for model building and posterior computation \citep{polson2015proximal,xu2024proximal,zhou2024proximal}.

The above geometric intuition helps explain the growing appeal of projected priors in Bayesian methodology. Projection onto a closed convex set is nonexpansive, so it preserves concentration rather than amplifying uncertainty, and this stability makes projection a natural transformation of a continuous latent variable. In favorable settings, the projection map is differentiable almost surely, which allows one to combine structured priors with gradient-based posterior computation. These ideas have already appeared in several forms in the literature. \citet{lin2014bayesian} used Gaussian-process projection to obtain Bayesian monotone regression under shape constraints.  \citet{xu2023l1ball} introduced the $\ell_1$-ball prior, which generates exact zeros via projection. This projection-based approach (and its variants) to inducing sparsity has been applied in various settings, including sparse vector autoregressive models \citep{hadjamar2024sparsevar}, soft-thresholded Gaussian processes  \citep{kang2018softthresholdgp}. \citet{xu2024proximal} generalized the same philosophy through proximal mappings, covering sparsity, fused structure, low-rank matrices, and other varying-dimensional spaces. \citet{thompson2024prodag} used projection-induced distributions to define priors supported on sparse directed acyclic graphs. Closely related ideas also appear in the broader literature on structured priors, including sparse Bayesian inverse problems \citep{everink2023regularizedgaussian}, image-on-scalar regression \citep{zeng2024imageonscalar}, and changepoint detection for topological image series \citep{thomas2025changepoint}. Together, these developments suggest that projected priors offer a practical route to combining structural regularization, uncertainty quantification, and gradient-guided posterior computation.

Despite these advantages, the practical usefulness of a projected prior depends critically on the cost of computing the projection, which is itself typically an optimization problem. When the projection map admits a closed form or is solvable by a non-iterative algorithm, posterior computation can be relatively straightforward. But in many cases, the projection is available only as the solution of a complicated distance-minimization problem, requiring iterative numerical routines at each evaluation. For example, when the set is formed by more than one constraint, such as the matrix space under both sparsity and positive semidefinite constraints, the projection does not have a non-iterative solution. Instead, the standard solution is to rely on alternating projection \citep{bauschke1996projection} or Dykstra's algorithm \citep{boyle1986method}, which requires multiple iterations to converge. Although convergence is generally rapid for the calculation of one projection, the computational costs add up quickly because the projection may need to be recomputed once or more in each Markov chain Monte Carlo iteration. In addition to computing the projection, numerical differentiation of the projection map can also be burdensome. Automatic differentiation through iterative solution, known as unrolled differentiation, is possible, but comes with a high computational cost unless the number of optimization iterations is small \citep{scieur2022curse}. This challenge is similar to Bayesian variable selection, where exact approaches like spike-and-slab priors yield sparsity but are computationally demanding, leading to the popularity of continuous shrinkage priors as efficient alternatives \citep{park2008bayesian,carvalho2009horseshoe,polson2010shrink}. Likewise, exact projected priors put mass on structured subsets but can be costly to compute. This motivates relaxing exact projection to a continuous prior that strongly favors points near the projection.

\section{Method}
\subsection{Constrained set, projection, and duality}
To motivate our framework, we first consider a model with likelihood function $\mathcal L(y;\theta)$, where $y$ denotes the observed data and $\theta\in \mathbb{R}^p$ is the parameter of interest. Suppose there is a subset $\mathcal M\subset \mathbb{R}^p$, representing the set of parameter values that are desirable for purposes of inference, prediction, or interpretation.

Typically, a constrained prior places positive probability only to points within $\mathcal M$ and zero probability to points outside $\mathcal M$. Here, we consider those constrained priors that assign positive probability for some constraints to become active (binding in the form of equality). To be illustrative, we start with the example of the regression variable selection problem, which ultimately inspires our solution.
In variable selection, the constrained space $\mathcal M$ can be expressed as
$$
\mathcal M = \bigcup_{S\subseteq [p]} \left\{ \theta : \theta_j = 0 \ \forall j \in S, \ \theta_k \neq 0 \ \forall k \in [p]\setminus S \right\},
$$
where $[p]=\{1,2,\ldots,p\}$. The union is $\mathbb{R}^p$ itself, however, positive prior probability is assigned to every subset on the right-hand side of the union, hence $\theta_j=0 \;\forall j \in S$ become active with a positive probability for some set $S$. If one assigns a binomial prior probability for each set, then the resulting prior is the discrete spike-and-slab prior.

For general problems, it is challenging to construct priors on sets with active constraints, but projection gives a viable solution. Suppose the subsets of $\mathcal M$ lie on the boundary of another set $C$. We may first draw $\beta$ from a continuous distribution such as a multivariate Gaussian, and then transform it into $\theta$ as the minimizer of a Euclidean distance:
\[\label{eq:euclidean-projection}
\arg\min_{z\in C} \frac12\| \beta - z\|_2^2 = \arg\min_{z}\frac12\| \beta - z\|_2^2 + I_{C}(z),
\]
which produces a Euclidean projection $\theta\in \mathcal M$ under suitable values of $\beta$, and $I_{C}(\theta)$ is the indicator function that takes value $0$ if $\theta\in C$ and $+\infty$ otherwise. It has been well established that when $C$ is a ball based on the $\ell_1$ norm $\sum_{j=1}^p |\theta_j|$, or some variant such as $\sum_{g=1}^G \sqrt{\sum_{j\in J_g} \theta_j^2}$, the projection does put positive probability mass on the actively constrained subsets of $\mathcal M$.

Generally, a canonical projected prior approach can be formulated as:
\begin{align}
\beta \sim \pi^\beta_0(\cdot), \quad \theta = T(\beta),
\label{eq:push}
\end{align}
where $\pi^\beta_0$ is a base prior distribution on the auxiliary variable $\beta$ and $T(\beta)$ is the projection of $\beta$ onto $C$. The transformation $T$ serves as a change of variables that produces a pushed-forward measure $T_{\#}\pi^\beta_0$ on $\theta$, which we refer to as the \emph{projected prior}. \eqref{eq:push} defines a distribution over the parameter of interest $\theta$ by deterministically transforming variable $\beta$, a flexible approach for constructing distributions without requiring an explicit density. Specifically, we take the transformation $T(\beta)$ to be the solution to an optimization problem parameterized by $\beta$. The Euclidean projection \eqref{eq:euclidean-projection} is a special case, and we discuss other forms of $T$ later.

The challenge in adopting the projection prior is that the projection $T$ may not be a tractable optimization solution, meaning that it is not solvable in closed form or by a non-iterative algorithm.
Therefore, we first consider the dual form of the projection and quantify the exact projection via the dual gap.
Since $I_{C}(z) =  \sup_{u \in \mathbb{R}^p} \big( u^\top z - \sigma_{C}(u) \big)$ with $\sigma_{C}(u) = \sup_{w\in C} u^\top w$ known as the support function of $C$,
we have the lower bound of projection loss using the max-min inequality,
\[\label{eq:projection_duality}
  & \frac12\| \beta - z\|_2^2 + I_{C}(z) \ge \min_{z \in \mathbb{R}^p}\frac12\| \beta - z\|_2^2 + I_{C}(z)\ge
  \sup_{u \in \mathbb{R}^p}\inf_{z \in \mathbb{R}^p}\frac12\| \beta - z\|_2^2 + u^\top z - \sigma_{C}(u) \\
  &\ge
  \beta^\top u - \frac12 \|u\|^2 - \sigma_{C}(u).
\]
The rightmost of the inequality is known as the dual function of $u\in\mathbb{R}^p$, given $\beta$; accordingly, the leftmost is known as the primal function of $z\in\mathbb{R}^p$, given $\beta$.

The above is commonly known as weak duality, and it holds for any pair of values $(z,u)\in \mathbb{R}^p\times \mathbb{R}^p$, regardless of whether $C$ is convex. Suppose we do not have a tractable solution for the projection problem, but somehow encounter a pair $(\hat z,\hat u)$ such that $(1/2)\| \beta - \hat z\|^2 + I_{C}(\hat z) = \beta^\top \hat u -(1/2) \|\hat u\|^2 - \sigma_{C}(\hat u)$ exactly and the common value is finite. Then every inequality must turn into an equality, and hence $\hat z$ is the exact projection. The difference between the primal and dual functions, or the \emph{duality gap}, characterizes the closeness of $\hat z$ to the projection.

In \eqref{eq:projection_duality}, even though $\hat u$ may seem redundant for our purpose, it computes the dual function as the lower bound and  {\em certifies} the optimality of $\hat z$. In optimization, there are algorithms, namely primal-dual methods, that do not operate directly by minimizing only the primal function, but instead simultaneously update $z$ and $u$ to reduce the duality gap. This motivates us to take a similar approach in Bayesian prior specification.

\subsection{Gap-shrinkage prior}
We consider a general projection problem of the form
\(
\theta = T(\beta)= {\arg\min}_{z\in C}\ f(z;\beta)
\)
where $\beta\sim \pi^\beta_0(\cdot)$, and we assume the minimizer is unique for almost every $\beta$. The family of transformations $T$ include Euclidean projections, proximal mappings, and Bregman projections. We call $f$ the primal loss, and assume it has an associated minorant $\phi$ and a dual function $d$:
\(
f(z;\beta) = \sup_{u\in\mathcal U} \phi(z,u;\beta), \quad z\in C, \qquad
d(u;\beta) = \inf_{z\in C} \phi(z,u;\beta), \quad u\in \mathcal{U},
\)
in which taking the infimum often yields a dependent relationship among $u$, $z$, and $\beta$.

The choice of $\phi$ is not unique, and depends on the optimization technique being used, such as variable splitting, Lagrange multipliers, or Fenchel conjugacy. We discuss the choice of $\phi$ in more detail in a later section. Regardless, a practical goal is to make $d$ tractable, which in turn yields a tractable gap function
$f(z;\beta)-d(u;\beta),
$
as an upper bound on $f(z;\beta)-\min_{z\in C} f(z;\beta)$. Exponentiating the negative of the gap function, we obtain a shrinkage prior:
\[\label{eq:primal-dual-prior}
  \Pi_0(\theta,u,\beta) & \propto  {\exp\left [ - \alpha \left\{ f(\theta;\beta) - d(u;\beta) \right\} \right ] } \pi^\beta_0(\beta),
  \]
  where $\theta\in C$ and $u\in \mathcal U$. We refer to \eqref{eq:primal-dual-prior} as the \emph{gap-shrinkage prior}, and $\alpha>0$ is a fixed hyperparameter that controls the strength of the shrinkage effect. Since the duality gap is always non-negative, given $(u, \beta)$, the gap-shrinkage prior places high density close to the exact projection where the gap equals zero.

  \begin{remark}
    Note that $(u,z)$ can be chosen to be independent or dependent in the prior specification, and it affects neither the weak duality nor the certification of the optimality of $\hat z$. In several cases presented in the article, we may choose $u$ as some deterministic transform of $\beta$ and $z$.
  \end{remark}
 \begin{figure}[H]
    \centering
    \includegraphics[width=0.45\textwidth]{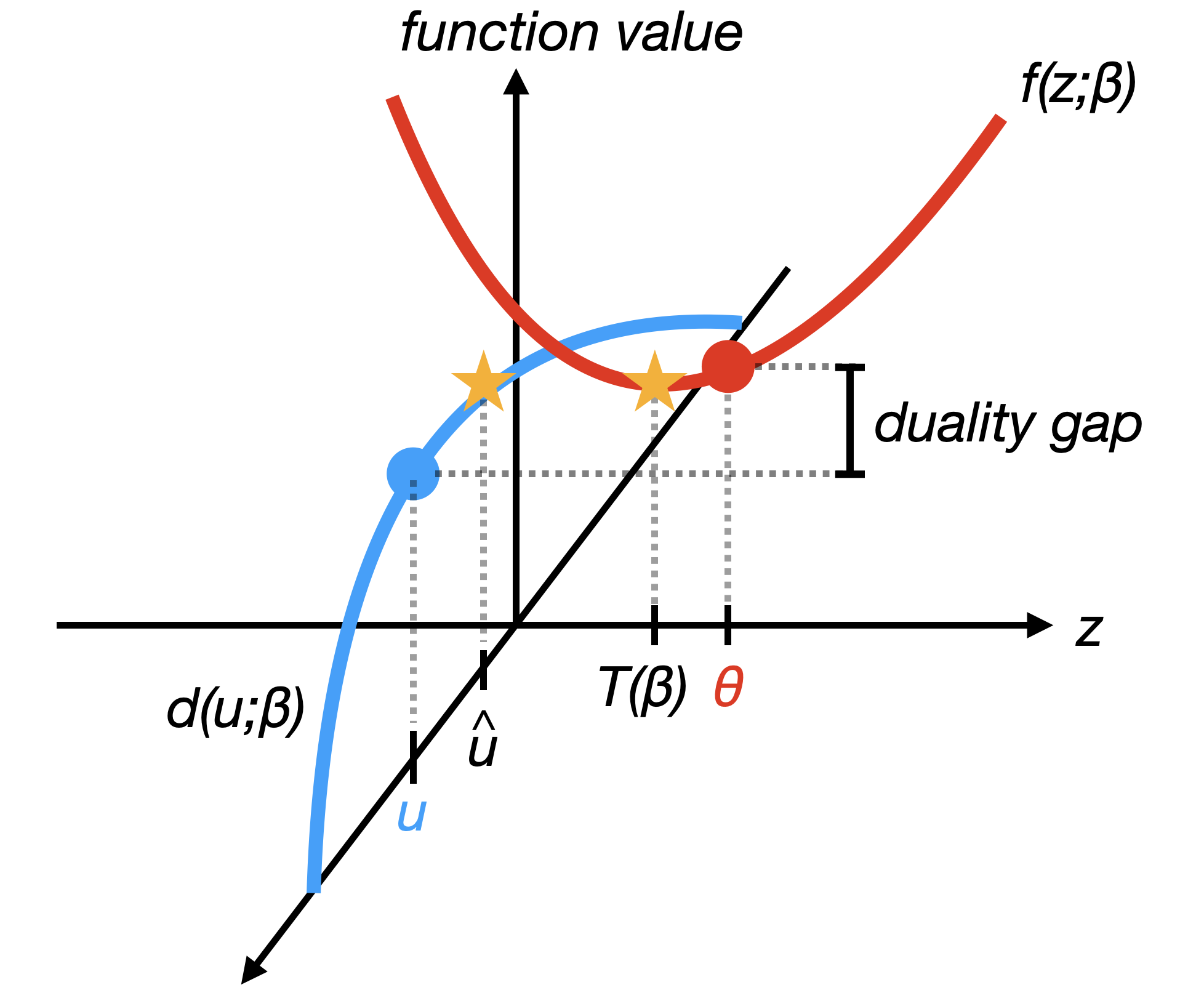}
    \caption{Illustration of gap-shrinkage prior. Let $f$ be the projection loss at a given $\beta$. The minimizer $T(\beta)$ (star on red curve) lives in a set of interest $C$, but is challenging to find. We use an augmented $u$ in the dual function $d(u;\beta)$, which has its maximum (star on blue curve) equals the minimum of $f(z;\beta)$. Since both $f$ and $d$ are tractable, a prior penalizing their gap makes $\theta$ close to $T(\beta)$.}
  \label{fig:duality_gap_illu}
\end{figure}
Beyond the Euclidean projection \eqref{eq:euclidean-projection}, we focus on a generalized projection known as a proximal mapping, which is the minimizer of the following proximal loss:
\[\label{eq:proximal-loss}
f(z;\beta) = \frac12\| \beta - z\|_2^2 + g(z),
\]
where $g$ is a convex and lower semi-continuous function. The minimizer $\arg\min_z f(z;\beta)$ is the proximal mapping, denoted by $\text{prox}_g(\beta)$. The above loss includes the Euclidean projection loss as a special case when $g(z)=I_C(z)$, although more often one uses a penalty form of $g$ instead of $I_C(z)$. Specifically, for Euclidean projection onto the set $C=\{z: h_j(z)\le r_j, j=1,\ldots,d\}$, where each $h_j:\mathbb{R}^p\mapsto \mathbb{R}$ is convex, one can show that for every $(r_1,\ldots,r_d)$ the Euclidean projection is also the minimizer of a proximal loss with penalty function $g(z)=\sum \lambda_j h_j(z)$ for suitable $\lambda_j\ge 0$. The preference for the penalty form over $I_C(z)$ is often driven by its greater computational simplicity. We first illustrate the gap-shrinkage prior with a simple $\ell_1$-norm proximal mapping.
\begin{example}
  Consider the $\ell_1$-norm proximal loss commonly used for inducing sparsity:
 \(
f(z;\beta,\lambda)= \frac12\|z-\beta\|_2^2+\lambda\|z\|_1,
 \)
 with $\lambda\ge 0$. In fact,
 the primal form has a closed-form minimizer  $\hat z = \text{sign}(\beta) (|\beta| - \lambda )_+$, which was used in the construction of $\ell_1$-ball prior \citep{xu2023l1ball} and soft-thresholded Gaussian process prior \citep{kang2018softthresholdgp}. Nevertheless, interesting continuous priors can be derived via the gap-shrinkage route. Using the dual norm of $\ell_1$-norm, we can see $\lambda\|z\|_1=\max_{\|u\|_\infty\le \lambda}u^\top z$, leading to minorant $\phi(z,u;\beta) = (1/2)\|z-\beta\|_2^2+ u^\top z$. Minimizing over $z$ yields a deterministic relationship $\hat z =\beta-u$, and the dual function:
\(
  d(u;\beta,\lambda)=  u^\top\beta-\frac12\|u\|^2, \ {\|u\|_\infty\le \lambda}.
\)
Setting $\theta=\hat z$ as the minimizer and applying $\beta=\theta+u$ to simplify, we have the gap function:
\(
f(\hat{z};\beta)-d(u;\beta) = \lambda\|\theta\|_1 - u^\top \theta.
\)
 We note that each element in the gap $\lambda|\theta_j|-u_j\theta_j$ is smaller when $u_j\theta_j>0$ than when $u_j\theta_j<0$; hence we can further impose constraint $u_j\theta_j>0$ for all $j=1,\ldots,p$.
Assigning a Cauchy kernel as the base prior $\pi^\beta_0(\beta)$, we have the following simplified gap-shrinkage prior:
\(
\Pi_0(\theta,u \mid \lambda) & \propto  {\exp\left\{ -
\alpha\sum_{j=1}^p (\lambda-|u_j|)|\theta_j|
\right\}}   1(\|u\|_\infty\le \lambda)   \prod_{j=1}^p  \{1+  (\theta_j+u_j)^2 \} ^{-1}.
\)
We can see resemblance to the Bayesian lasso prior $\propto\exp(- \alpha \lambda \|\theta\|_1)$, but also a clear difference in the $\alpha |u_j||\theta_j|$ term.  Effectively, a large $\alpha$ makes the following happen with high probability:
\(
  (\lambda-|u_j|)|\theta_j| \approx 0 \quad \Leftrightarrow \quad \theta_j\approx 0 \text{ or } \lambda \approx |u_j|.
\)
There is a connection to the global-local shrinkage literature \citep{polson2010shrink}: $\lambda \gg |u_j|$ leads to a strong global shrinkage effect driving $\theta_j$ toward the zero (active constraint), while the dual parameter $u_j$ provides a local adjustment allowing $\theta_j$ to escape from zero as $|u_j|\uparrow \lambda$.
\end{example}
\begin{remark}
  In the supplementary material, we prove that under Cauchy kernel for $\pi^\beta_0$, the marginal prior $\pi_0(\theta_j | \lambda)$ has a power-law tail, instead of exponential as in the Bayesian lasso.
\end{remark}

For generalized $\ell_1$-norm proximal losses of the form $g(z) = \lambda \|Dz\|_1$,
the primal loss function does not have a tractable solution. Yet such losses have important applications in graph-based smoothing, change-point detection, and other structured sparsity problems. For example, if a list of contrasts $(i,j)\in E$ is represented as edges in a graph $G=([p],E)$, then $D$ can be taken as the corresponding contrast design matrix with values in $\{-1,0,1\}^{|E|\times p}$. We now derive the associated gap-shrinkage prior.

\begin{example}
  Consider the proximal loss
  $
    f(z;\beta,\lambda)=\frac12\|z-\beta\|_2^2+\lambda\|Dz\|_1,
  $
  for $D\in \mathbb{R}^{d\times p}$.
  Using a variable splitting technique $Dz=w$ and Lagrange multiplier $u\in \mathbb{R}^d$, we have
\(\phi(z,w,u;\beta) = \frac12\|z-\beta\|_2^2+\lambda\|w\|_1 + u^\top (D z-w).\)
 Minimizing over $z$ and $w$ yields relationship $\theta=\hat z =\beta - D^\top u$ and $\min_w \lambda \|w\|_1 - u^\top w=0$ if $\|u\|_\infty \le \lambda$ (otherwise, the infimum is $-\infty$), hence the dual function is:
  \(
  & d(u;\beta,\lambda)=-\tfrac12\|D^\top u\|_2^2+\beta^\top D^\top u, \quad \|u\|_\infty\le\lambda.
  \)
  Using a similar constraint $u_j (D\theta)_j>0$ and assigning a Cauchy kernel $\pi^\beta_0$, replacing $z$ in primal loss function with the minimum $\theta$, we have a gap-shrinkage prior:
\(
\Pi_0(\theta,u \mid \lambda) & \propto {\exp\left\{ - \alpha   \sum_{j=1}^d (\lambda - |u_j|) |(D\theta)_j| \right\}}   1(\|u\|_\infty\le \lambda)   \prod_{j=1}^p  \{1+  (\theta_j+ D_j^\top u)^2 \} ^{-1}.
\)
Note a similar adaptive shrinkage effect on $ D\theta$ as in the last example.
\end{example}

\subsection{Properties of gap-shrinkage}
We now derive some useful properties of gap-shrinkage priors. We first generalize the examples above to broad proximal mapping projection \eqref{eq:proximal-loss}.

\medskip\noindent\textit{Tractable relaxation of proximal mapping projection.}\quad
Our goal is to derive a tractable form of the gap-shrinkage prior for proximal mapping \eqref{eq:proximal-loss}, so that readers can readily apply it to their choice of $g$. Using variable splitting $z=w$, we have the Lagrangian for \eqref{eq:proximal-loss}:
\(
  \phi(z,w,u;\beta) = \frac12\|\beta - z\|^2 + g(w) + u^\top (z - w),
\)
whose minimization over $(z,w)$ yields $z = \beta - u$, $u\in \mathbb{R}^p$, and the dual function:
$
  d(u;\beta) = u^\top \beta - \frac12 \|u\|^2 -g^*(u),
$
where $g^*(u) = \sup_{w} \big\{u^\top w - g(w)\big\}$ is known as the Fenchel conjugate of $g$.

Therefore, we have the gap function under $z=\beta-u$:
\[\label{eq:gap-function-proximal}
  f(z;\beta)-d(u;\beta)& = \frac12\|\beta - z\|^2 + g(z) -  u^\top \beta + \frac12 \|u\|^2 +g^*(u)\\
  & =    g^*(u)+ g(z) -   u^\top z  \big\vert_{z=\beta-u}.
\]
We already know the non-negativity of the left hand side. In the meantime, the non-negativity of the right hand side is also famously known as the Fenchel--Young inequality, and trivially provable using the definition $g^*(u)$ as a supremum. As a result, the gap-shrinkage is
\(
  \Pi_0(\theta,u\mid \lambda) & \propto  {\exp\left [ - \alpha \left\{ g^*(u)+ g(\theta)-   u^\top \theta \right\} \right ] } \pi^\beta_0(\theta+u),
\)
where we replace $z$ in the primal function with the minimizer $\theta$ and apply $\theta = z = \beta-u$ (i.e., $\beta = \theta + u$) in $\pi^\beta_0$. In practice, many Fenchel conjugates admit closed-form expressions. Here are a few common examples. For $g(z) = \|z\|$ with respect to some norm $\|\cdot\|$, the conjugate is $g^*(u) = I_{\{ v : \|v\|_* \le 1 \}}(u)$, with $\|\cdot\|_*$ denoting the dual norm. When $g(z) = \frac{1}{2} z^\top Q z$ for a positive definite $Q$, the conjugate becomes $g^*(u) = \frac{1}{2} u^\top Q^{-1} u$. We note that the gap function in Examples 1 and 2 can be derived from \eqref{eq:gap-function-proximal} as well.

For more insights, the low-dimensional set $(\theta, u) \in \mathbb{R}^{p}\times \mathbb{R}^{p}$ has the following equivalence:
\[\label{eq:zero-gap-condition}
  g(\theta)+g^*(u)-\theta^\top u=0 \iff u \in \partial g(\theta) \iff \theta \in \partial g^*(u),
\]
where $\partial g(\theta)$ denotes subdifferential, i.e., the collection of all subgradients of $g$ at $\theta$. When the function is differentiable, subdifferential reduces to a single gradient element. If we view $\theta$ and $u$ as {\em two sides of the same coin}, then \eqref{eq:zero-gap-condition} tells us how to flip the coin to get the other side. We again use $g(\theta)=\lambda\|\theta\|_1$ to illustrate: $\partial g(\theta)=\{u: u_i= \lambda\text{sign}(\theta_i), \text{ if } \theta_i\neq 0, u_i\in [-\lambda,\lambda] \text{ if } \theta_i= 0\}$, $g^*(u)=I_{ \{w:\|w\|_\infty\le \lambda\}}(u)$ has $\partial g^*(u)=\{ \theta: \theta_i=0 \text{ if } |u_i|<\lambda, \text{sign}(\theta_i)= \text{sign}(u_i) \text{ if } |u_i|=\lambda\}$.

Accordingly, for exact $\theta=T(\beta)$, the coupling between $\beta$ and $\theta$ gives an inverse-prox set
\(
\mathbb M = \{(\beta,\theta)\in \mathbb{R}^p \times \text{dom}(g): \beta =\theta + u, u\in \partial g(\theta)\}.
\)
The above further reduces to $\beta =\theta + \nabla g(\theta)$ (the inverse of proximal mapping), if $g$ is differentiable at $\theta$.
Therefore, we see the gap-shrinkage $\Pi_0(\beta,\theta)$ is a continuous prior in $\mathbb{R}^p \times \text{dom}(g)$, with mass concentrated near the $\mathbb M$. Figure \ref{fig:manifold-example} shows the geometric intuition.

\begin{figure}[H]
  \begin{subfigure}[t]{0.45\textwidth}
  	\centering
    \includegraphics[width=0.7\textwidth]{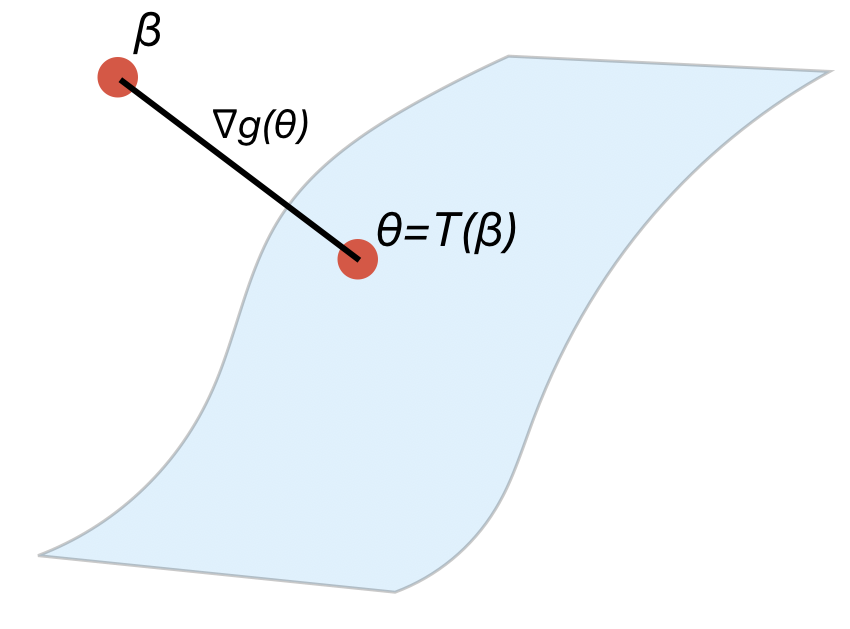}
    \caption{Illustration of the inverse-prox set $\mathbb M$ for $(\beta,\theta)$, $\beta=\theta + \nabla g(\theta)$ for a differentiable $g$. If $g$ is not differentiable at $\theta$, $\nabla g(\theta)$ is replaced with subgradient.}
  \end{subfigure} \qquad
  \begin{subfigure}[t]{0.45\textwidth}
    	\centering
    \includegraphics[width=0.7\textwidth]{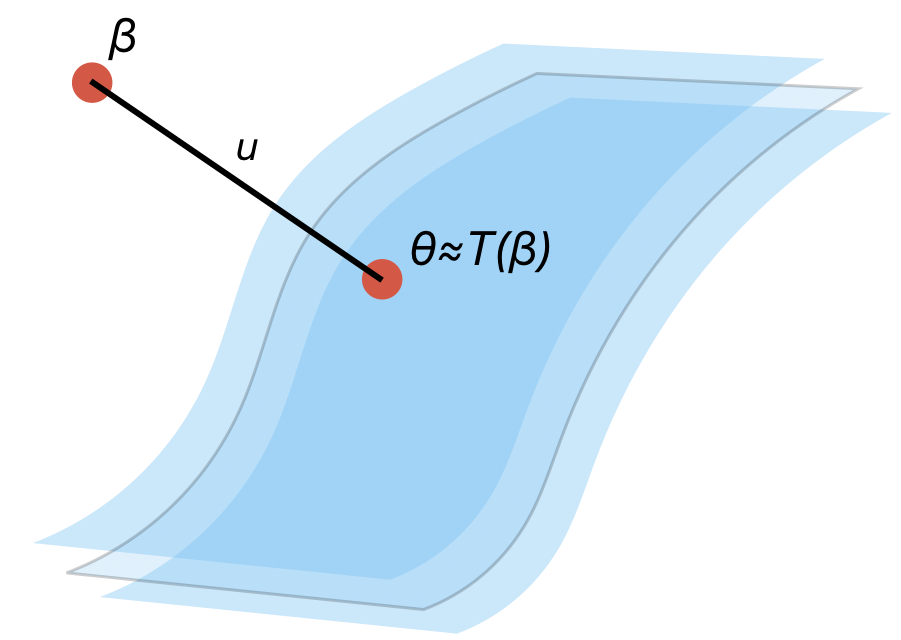}
    \caption{Gap-shrinkage prior relaxes the coupling to $\beta=\theta + u$, so that $\theta$ no longer resides on the low-dimensional set, but $\theta \approx T(\beta)$.}
  \end{subfigure}
  \caption{Illustration of the inverse-prox set for $(\beta,\theta)$ under proximal mapping, and the relaxation under gap-shrinkage prior. In the latter, $\theta \approx T(\beta)$ thanks to the distance bound via the gap function.}
  \label{fig:manifold-example}
  \end{figure}
\noindent\textit{Distance bound via the gap function.}\quad Next, we establish the effectiveness of shrinkage. Specifically, since $T(\beta)$ is intractable and we can only quantify $f(z;\beta)-d(u;\beta)$, how far is the relaxed $z$ from the exact $T(\beta)$? We can now use the following theorem to quantify the difference.

\begin{theorem}\label{thm:shrinkage-effect}
  If $f(z;\beta)$ is $\mu$-strongly convex, i.e., $f(z;\beta)- (\mu/2)\|z\|_2^2$ is convex,  its unique minimizer $\hat z=\arg\min_z f(z;\beta)$,
  then for any $z\in C$ and $u\in \mathcal U$,
  \(\|z - \hat z\|_2 \;\le\; \sqrt{({2}/{\mu}) \left\{ f(z;\beta)-d(u;\beta) \right\}}.\)
\end{theorem}
\noindent The proof is based on a direct application of strong convexity.
The theorem reveals that encouraging a small duality gap effectively bounds the deviation of the relaxed $z$ from the exact $T(\beta)$. We can immediately see that proximal mapping (including Euclidean projection to convex set as a special case) with a convex $g$ is $1$-strongly convex, hence
$\|z - T(\beta)\|_2 \;\le\;  \sqrt{2 \{ f(z;\beta)-d(u;\beta)\}}.$

\medskip\noindent\textit{Generalizability to broad projection.}\quad
To obtain more general results applicable to a broad class of projections, we introduce the Bregman projection, which minimizes the \textit{Bregman divergence}:
\[
\label{eq:bregman-divergence}
D_\psi(z,\beta) = \psi(z) -\psi(\beta) - ( z-\beta)^\top \nabla \psi(\beta),
\]
with $\psi:\mathbb{R}^p\to \mathbb{R} \cup \{+\infty\}$ an extended-value function that is continuously differentiable and strictly convex. We can obtain Bregman projection as $T(\beta)= {\arg\min}_{z\in \mathbb{R}^p} D_\psi(z,\beta) + I_C(z).$ The Bregman divergence $D_\psi$ is a generalization of the squared Euclidean distance (when $\psi(z)=\|z\|_2^2$). The Bregman divergence $D_\psi$ may be asymmetric and need not satisfy the triangle inequality, but it provides a natural notion of discrepancy tailored to the geometry induced by $\psi$. We note that strong convexity is not necessary if we want to have effective control on the Bregman divergence between $z$ and $T(\beta)$, as shown in the following theorem.
\begin{theorem}\label{thm:bregman-shrinkage-effect}
  If $D_\psi$ is a Bregman divergence generated by a continuously differentiable, strictly convex function $\psi:\mathbb{R}^p \to \mathbb{R} \cup \{+\infty\}$, then for the primal $f(z;\beta)=D_\psi(z,\beta)+I_C(z)$, the dual function is
$
d(u;\beta)
=
-\psi^*(\nabla\psi(\beta)-u)-\sigma_C(u)
+\beta^\top \nabla\psi(\beta)-\psi(\beta),
 $
with $u\in \mathbb{R}^p$.
Let the unique minimizer be $\hat z=\arg\min_z D_\psi(z,\beta)+ I_C(z)$. Then, for any $z\in C$ and $u\in \mathbb{R}^p$,
  \(D_\psi(z,\hat z) \;\le\;   f(z;\beta)-d(u;\beta) .\)
\end{theorem}
\noindent The theorem  applies to $\psi$ with domain $\mathbb{R}^p$. For those $D_{\tilde\psi}$ generated by $\tilde\psi$ with domain $\Omega \subseteq \mathbb{R}^p$, one could use $\psi(z)=\tilde\psi(z) + I_{\Omega}(z)$ to reconcile the domain difference, provided the divergence is calculated for points where the gradient exists.
We now present an example.
\begin{example}
  Consider negative entropy
  $
  \psi(z)=\sum_{j=1}^p z_j\log z_j + I_{\Delta^{p-1}}(z)$ for $z\in \mathbb{R}^p$, whose Bregman divergence is KL divergence
  $
  D_\psi(z,\beta)
  =
  \sum_{j=1}^p z_j\log({z_j}/{\beta_j})
  $
  Let $C_0\subset\mathbb{R}^p$ be a convex set and let $\Delta^{p-1}=\{z: z\ge 0, \sum_{j=1}^p z_j = 1\}$ denote the probability simplex (which is convex). We have a Bregman projection from $\beta \in \Delta^{p-1}$ to a subset within:
  $
  T(\beta)={\arg\min}_{z\in C_0\cap\Delta^{p-1}} \sum_{j=1}^p z_j\log({z_j}/{\beta_j}).
  $
  Compared to the Euclidean projection that could produce $z_j=0$, the above preserves positivity, which is often more natural for probability parameters.
  Since for $\psi(z)=\sum_{j=1}^p z_j\log z_j + I_{\Delta^{p-1}}(z)$,
the Fenchel conjugate is
$
\psi^*(u)=\log\Big(\sum_{j=1}^p e^{u_j}\Big).
$
We derive the dual function explicitly:
$
d(u;\beta)
=
-\log\Big(\sum_{j=1}^p \beta_j e^{-u_j}\Big)-\sigma_{C_0}(u),
$
where $\sigma_{C_0}(u)=\sup_{z\in C_0}\langle u,z\rangle$ is the support function of $C_0$. Therefore, the gap function takes the form
$
\sum_{j=1}^p z_j\log\frac{z_j}{\beta_j}
+\log\Big(\sum_{j=1}^p \beta_j e^{-u_j}\Big) + \sigma_{C_0}(u).
$
\end{example}

\subsection{Variational
gap function for complicated settings}
\label{sec:variational}
In complex settings, the gap function may be either not fully tractable or computationally expensive to evaluate. We now handle these cases by introducing a variational gap function:
\(
  \tilde G (z,v;\beta) \ge    f(z;\beta) - d(u;\beta), \quad \text{for all } z,u,v,
\)
and $\exists v=v(z,u,\beta)$ such that the equality can be achieved for any $(z,u)$. Therefore, we see the variational gap function is a tight upper bound of the true gap function, hence allowing us to use $\exp[-\alpha \tilde G (\theta,v;\beta)]$ for gap-shrinkage.

To be more concrete, we now discuss two applications of the variational gap function. In the first case, in proximal mappings, it is common to have an additive $g(z)=\sum_{j=1}^J g_j(z)$, such as a sum of different regularizers or cost functions, where each $g_j$ has a tractable Fenchel conjugate. We can use the Fenchel conjugate for the additive functions given by infimum convolution:
\begin{align}
  g(z)=\sum_{j=1}^J g_j(z) \quad \text{then} \quad g^*(u)=\inf_{\substack{ \sum_{j=1}^J v_j=u}}\sum_{j=1}^J g_j^*(v_j).
  \label{eq:variational}
\end{align}
The infimum convolution may not have a closed-form solution, motivating for a variational gap:
\begin{align}
  \tilde G (z,v;\beta)= \sum_{j=1}^J g_j^*(v_j)+ g(z)-   \sum_{j=1}^J v_j^\top z \big\vert_{z=\beta - \sum_{j=1}^J v_j},
\end{align}
corresponding to $v=(v_1,\dots,v_J)$, each $v_j\in \mathbb{R}^p$; $\tilde G (z,v;\beta)$ is a reachable upper bound of the  gap because $\sum_{j=1}^J g_j^*(v_j) \geq g^*(u)$ by \eqref{eq:variational}, which can then be plugged into the gap function in \eqref{eq:gap-function-proximal}.

\begin{example}
  Consider the projection to a set $C=\bigcap_{j=1}^J C_j$, with each $C_j$ a convex set. Using $g_j(z)=I_{C_j}(z)$, we can see that $g(z)=\sum_{j=1}^J I_{C_j}(z)$ is equivalent to $I_C(z)$. Therefore, the additive form of $g$ corresponds to the projection to the intersection of the sets.
  Since $I_{C_j}(z)$ has Fenchel conjugate $\sigma_{C_j}(v_j) $, the support function of $C_j$, we have the gap-shrinkage prior given by
  \(
  \Pi_0(\theta,v_1,\ldots,v_J\mid \lambda) & \propto  {\exp\left [ - \alpha \sum_{j=1}^J \left\{  \sigma_{C_j}(v_j) -   v_j^\top \theta \right\} \right ] } \pi^\beta_0(\theta+ \sum_{j=1}^J v_j) 1(\theta\in C_j).
  \)
  For simple convex sets $C_j$, the support function $\sigma_{C_j}(v_j)$ has a closed-form expression. For example, for a ball $C_j=\{z: \|z\|\le r\}$, generated by some norm $\|\cdot\|$, we have $\sigma_{C_j}(v_j) = r\|v_j\|_*$.

  For a concrete application, for group-wise sparsity, we may consider a projected prior generated by the intersection of the sets $C_j=\{z: \sqrt{\sum_{i\in G_j} z^2_i} \le r_j\}$, where $G_j$ is a group of indices that we want the elements of $\theta_j$ to be simultaneously zero or non-zero. The associated support function is $\sigma_{C_j}(u_j) = r_j \sqrt{\sum_{i\in G_j} u^2_{j,i}}$, since the dual norm of the $\ell_2$ norm is another $\ell_2$ norm.
\end{example}

The second application involves a matrix-valued parameter $\theta\in \mathbb{R}^{p_1\times p_2}$, whose operations may be expensive. For example, for projection to a low-rank space, one often needs to apply some matrix decomposition such as singular value decomposition (SVD), which is of cubic order $O(p_1 p_2 \min(p_1,p_2))$. This motivates us to develop a variational gap function.
\begin{example}
  Let $\theta\in\mathbb{R}^{p_1\times p_2}$ be a matrix parameter and
  $g(\theta)=\lambda\|\theta\|_* $, where $\|\cdot\|_*$ denotes the nuclear norm, the sum of the singular values of $\theta$.   The Fenchel conjugate of $g$ is the indicator of the operator-norm ball,
  $g^*(u)=I_{\{\|u\|_{\mathrm{op}}\le\lambda\}}(u)$, with $u\in \mathbb{R}^{p_1\times p_2}$.
  Substituting into \eqref{eq:gap-function-proximal} gives
  \(
    f(z;\beta)-d(u;\beta)
    \;=\;
    \lambda\|\theta\|_* - \langle u,\,\theta\rangle,
    \qquad \|u\|_{\mathrm{op}}\le\lambda,\quad \theta=\beta-u.
  \)
  The proximal mapping with respect to this $g$ has a closed-form solution given by singular-value soft-thresholding of $\theta$, but this approach requires computing the SVD of $\theta$. The same burden applies to evaluating $\|\theta\|_*$.
Instead, we use the variational form:
  \(
    \|\theta\|_*
    \;=\;
    \min_{\substack{A\in\mathbb{R}^{p_1\times r},\;B\in\mathbb{R}^{p_2\times r}\\AB^\top=\theta}}
    \frac{1}{2}\bigl(\|A\|_F^2+\|B\|_F^2\bigr),
  \)
  which holds for any $r= \min(p_1,p_2)$. After reparameterizing, we have a modified  prior
  \(
    \Pi_0(\theta=AB^\top,u\mid\lambda)
    \;\propto\;
    \exp\!\left[-\alpha\left\{
      \frac{\lambda}{2}\bigl(\|A\|_F^2+\|B\|_F^2\bigr)
      -\langle u,\,AB^\top\rangle
    \right\}\right]
    \pi_0^{\beta}(AB^\top+u),
  \)
  where $\|u\|_{\mathrm{op}}\le\lambda$.
\end{example}

\subsection{Posterior under gap-shrinkage}

We denote the likelihood as $\mathcal{L}(y;\theta,\kappa)$ with $y$ the data and $\kappa$ some other parameters (such as $\lambda$ used above) associated with a prior $\Pi_0(\kappa)$, then our posterior is given by
\(
\Pi(\theta,u,\kappa\mid y) =
\frac{
\mathcal{L}(y;\theta,\kappa) \Pi_0(\theta,u,\beta\mid \kappa) \Pi_0(\kappa)
}{
\int \mathcal{L}(y;\theta,\kappa) \Pi_0(\theta,u,\beta\mid \kappa) \Pi_0(\kappa) \textup{d}(\theta,u,\beta,\kappa).
}
\)
In this article, we focus on the case where the numerator is fully tractable, except for some normalizing constant invariant to $(\theta,u,\beta,\kappa)$.  Since the quantity of interest is often only the marginal posterior of $\theta$ and $\kappa$, one can collect estimates of the joint posterior in the augmented space, then discard the information of $u$ \citep{tanner1987calculation}. On the other hand, as we discussed above, having information about $u$ can be useful for quantifying the difference between $\theta$ and the exact projection $T(\beta)$.

 In the theory section, we will examine the large sample properties of the posterior under the gap-shrinkage prior carefully. For now, we conduct a simple analysis by examining the Hessian of the log-posterior, whose inverse estimates the asymptotic covariance of the parameters. For simplicity, we focus on the proximal mapping case now.
$$
\Pi(\theta,u,\kappa\mid y) \propto
\underbrace{\mathcal{L}(y;\theta,\kappa)  \pi^\beta_0(\theta+u)\Pi_0(\kappa)}_{=:\exp(-V(\theta,u,\kappa))} \exp [-\alpha \{ (g(\theta)+g^*(u) - u^\top \theta)\}].
$$
The observed negative Hessian given the data is
\begin{align}
J(\theta,u,\kappa)=
\nabla^2 V(\theta,u,\kappa)+
\alpha  \begin{bmatrix}
\nabla^2 g(\theta) &
- I &
0\\
  - I &
 \nabla^2 g^*(u) &
0\\
0 &
0 &
0
\end{bmatrix}.
\label{eq:hessian}
\end{align}
By \eqref{eq:zero-gap-condition}, at the primal-dual optimal $(\hat\theta,\hat u)$ and if differentiable, we have $\nabla^2 g^*(\hat u)=\big(\nabla^2 g(\hat\theta)\big)^{-1}$. Therefore, at $(\theta,u)$ near $(\hat\theta,\hat u)$, we have $\nabla^2 g(\theta) \nabla^2 g^*(u)-I \approx 0$. As a result, at a large $\alpha$, the first $2\times 2$ block of $J(\theta,u,\kappa)$ will be dominated by a close-to-singular matrix, corresponding to a nearly deterministic relation between $\theta$ and $u$.
\section{Posterior computation}
\subsection{Hamiltonian Monte Carlo}
An advantage of the Gap-shrinkage prior is that all parameters are continuous and the posterior density is tractable up to some constant.
 Therefore, Hamiltonian Monte Carlo (HMC) \citep{neal2011mcmc} is particularly suitable for posterior computation. We provide a brief overview of the algorithm. HMC augments $\xi$ with an auxiliary momentum variable $r\sim\mathrm{N}(0,M)$ and targets the joint density
$
\Pi(\xi,r\mid y) \propto \exp\bigl\{-\mathcal{H}(\xi,r)\bigr\}, \; \mathcal{H}(\xi,r) = U(\xi)+\tfrac{1}{2}r^\top M^{-1}r,
$
where $U(\xi)=-\log\Pi(\xi\mid y)$ is the potential energy. Hamilton's equations $\dot\xi = M^{-1}r$ and $\dot r = -\nabla_\xi U(\xi)$ define a volume-preserving, energy-conserving flow. One draws $r$ from the normal distribution, denotes the current state by $(\xi^{0},r^{0})$, and then simulates the flow to obtain a proposal $(\xi^{t},r^{t})$ after trajectory length $t$. If one could evaluate the flow exactly, then the energy-conserving property would ensure that $\mathcal H(\xi^{t},r^{t})=\mathcal H(\xi^{0},r^{0})$. Since in most cases the flow must be approximated numerically, typically by the leapfrog integrator:
$
r^{(1/2)} = r^{0}-\frac{\varepsilon}{2}\nabla_{\xi}U(\xi^{0}),\quad\xi^{1} = \xi^{0}+\varepsilon M^{-1}r^{(1/2)},\quad r^{1} = r^{(1/2)}-\frac{\varepsilon}{2}\nabla_{\xi}U(\xi^{1}),
$
 exact energy conservation no longer holds. Therefore, a Metropolis--Hastings step is used to accept the proposed state with probability $\min\Bigl\{1,\exp\bigl(-\mathcal H(\xi^{L},r^{L})+\mathcal H(\xi^{0},r^{0})\bigr)\Bigr\}.$ The No-U-Turn Sampler \citep{hoffman2014no} adaptively selects the trajectory length, removing the need to hand-tune the number of leapfrog steps.

 For gap-shrinkage prior under proximal loss, \eqref{eq:hessian} shows that the joint $(\theta,u)$ Hessian block becomes near-singular at large $\alpha$, which may degrade HMC performance. Fortunately, there is a nice remedy via reparameterizing $\Pi(\theta,u\mid y)$ into $\Pi(\beta,\theta\mid y)$ or $\Pi(\beta,u\mid y)$. In the former case,
 $$ \Pi_0(\theta \mid \beta) \propto \exp(-\alpha \{ \|\theta\|_2^2+ g(\theta)+g^*(\beta-\theta) - \beta^\top\theta  \}),$$
$g$ and $g^*$ are both convex, $\beta^\top\theta$ is linear in $\theta$, hence the conditional prior with $\|\theta\|_2^2$ is strongly log-concave in $\theta$, with prior $\text{Cov}(\theta\mid \beta) \preceq (2\alpha)^{-1} I$ \citep{saumard2014log}. The same result applies to $\Pi_0(u\mid\beta)$ as well. Since HMC is particularly effective for handling strongly log-concave distributions \citep{mangoubi2017rapid}, the conditional posterior $\Pi(\theta\mid\beta, \cdot)$ or  $\Pi(u\mid\beta, \cdot)$ is well-conditioned, and the effect of large $\alpha$ reduces to just having a small scale $1/(2\alpha)$. Therefore, we recommend using $(\beta,\theta)$ or $(\beta,u)$ as $\xi$ in HMC, instead of $(\theta,u)$. An alternative approach is to use the more general framework of Riemannian HMC \citep{girolami2011riemann}, which adapts to the local curvature of the potential.

\subsection{Metropolis-within-Gibbs sampling}
For some specific models, such as linear regression or generalized linear models, the Gibbs sampler can be more efficient than HMC in computation. Therefore, one could incorporate the existing Gibbs samplers, such as \cite{bhattacharya2016fast}, with an additional Metropolis-within-Gibbs step to update the dual parameters $u$. For conciseness, we omit the details in the main text, but provide an implementation for linear regression in the supplementary material.

\section{Theory on posterior consistency and contraction rate}
The gap-shrinkage prior in \eqref{eq:primal-dual-prior} is a relaxation of the exact projected prior. The exact projected prior is supported on the image of the projection map $T(\beta)$, while the gap-shrinkage prior assigns mass to $(\theta,u,\beta)$ with duality gap $\ge 0$, including values $\theta\neq T(\beta)$. This raises a natural concern: even though Theorems \ref{thm:shrinkage-effect} and \ref{thm:bregman-shrinkage-effect} show that a small gap implies that $\theta$ is close to the exact projection, as the sample size $n\to \infty$, does the posterior under the relaxed prior still concentrate at the true parameter $\theta_0$?

To address this question, we focus on the setting $Y_1,\ldots,Y_n\stackrel{\mathrm{iid}}{\sim} p_{\theta_0}$, where the true parameter $\theta_0$ lies in the projected set $\mathcal M=\{T(\beta):\beta\in \mathbb{R}^p\}$, including the possibility that $\theta_0$ is on the boundary of $\mathcal M$. In this section, $Y_i$ is synonymous with $y_i$, but we use the capital letter to emphasize that it is a random variable on which we describe the convergence in probability $P_{\theta_0}^n$. Suppose that under the exact projected prior, the posterior is consistent at $\theta_0$. Note that if the exact projected prior itself does not concentrate at $\theta_0$, then no comparison with its relaxation is meaningful.

\subsection{Posterior consistency}
We first show that the relaxed posterior is also consistent under a standard prior support condition. Let
$
G(\theta,u,\beta):=f(\theta;\beta)-d(u;\beta)\ge 0
$
denote the duality gap. Since $G(\theta,u,\beta)=0$ if and only if $\theta=T(\beta)$ and $u$ is dual optimal, the gap-shrinkage prior assigns positive mass to arbitrarily small neighborhoods of the exact projected set whenever its base prior does. Therefore, if the exact projected prior is consistent because it places sufficient mass near $\theta_0$, then the relaxed prior inherits the same local support near $\theta_0$. We now formally state the theorem.

    \begin{theorem}\label{thm:consistency-relaxed}
      Assume that $Y_1,\ldots,Y_n \stackrel{\mathrm{iid}}{\sim} p_{\theta_0}$, where $\{p_\theta:\theta\in\Theta\}$ is dominated by a common $\sigma$-finite measure, and let $\rho$ be a metric on $\Theta$. Suppose that the model satisfies the conditions of Schwartz's consistency theorem \citep{schwartz1965onbayes,barron1999consistency} using the metric $\rho$. Further assume that for every $\eta>0$, the gap-shrinkage prior assigns positive mass to
      $
      \Bigl\{(\theta,u,\beta):
      G(\theta,u,\beta)=0,\;
      KL(p_{\theta_0},p_\theta)<\eta
      \Bigr\},$ where
      $
      KL(p_{\theta_0},p_\theta)
      :=
      \int \log\left(\frac{p_{\theta_0}}{p_\theta}\right)p_{\theta_0}\,d\mu.
      $
      Then the posterior under the gap-shrinkage prior is consistent at $\theta_0$, that is, for every $\epsilon>0$, the posterior
      \(
      \Pi_n\!\left(\rho(\theta,\theta_0)>\epsilon \mid Y_{1:n}\right)\stackrel{P_{\theta_0}^n}{\to} 0.
      \)
      \end{theorem}

Theorem \ref{thm:consistency-relaxed} shows that the relaxation does not compromise posterior consistency so long as the gap-shrinkage prior preserves sufficient Kullback--Leibler support near the truth through its zero-gap component. This is particularly relevant when the true parameter $\theta_0$ lies on the boundary of $\mathcal M$: in such cases, nearby values outside $\mathcal M$ may be arbitrarily close in Euclidean distance, so exact support is not the essential issue for consistency.

\subsection{Contraction rate}
We now turn to contraction rates. Suppose that under the exact projected prior, the posterior contracts at rate $\varepsilon_n$ under the metric $\rho$, in the sense that for every sufficiently large $M>0$,$
\Pi^{\mathrm{ex}}\!\left(\rho\!\left(T(\beta),\theta_0\right)>M\varepsilon_n \mid Y_{1:n}\right)\to 0
$
in $P_{\theta_0}^n$-probability.
For the relaxed posterior, the concern is that the additional discrepancy between $\theta$ and $T(\beta)$ may deteriorate this rate. Intuitively,
$$
\rho(\theta,\theta_0)\le \underbrace{\rho(\theta,T(\beta))}_{\text{relaxation error}}+\underbrace{\rho(T(\beta),\theta_0)}_{\text{projection error}}.
$$
Therefore, the contraction rate could be established if both terms on the right-hand side are of $\varepsilon_n$ scale. In the following, we proceed in two steps by first showing $G(\theta,u,\beta)$ is of $\varepsilon_n^2$ scale hence enabling the possibility of $\rho(\theta,T(\beta))\downarrow 0$ at an $\varepsilon_n$ rate by Theorem~\ref{thm:shrinkage-effect}; we then establish the sufficient conditions for $\rho(T(\beta),\theta_0)\downarrow 0$ at the same rate.
First, the  prior promotes small gaps hence small $\rho(\theta,T(\beta))$ as {\em a priori}, but it does not  guarantee rate $\varepsilon_n$ for $\rho(\theta,T(\beta))$ in the posterior. Therefore, we need some technical conditions.

\begin{theorem}\label{thm:distance-relaxed}
Consider the following conditions.
 \begin{enumerate}
\renewcommand{\labelenumi}{(\roman{enumi})}
  \item There exists a sequence $\varepsilon_n \to 0$ such that $n\varepsilon_n^2 \to \infty$.

  \item There exists a sieve $\mathcal F_n \subseteq \Theta \times \mathcal U \times \mathbb R^p$ such that
  $
  \Pi_n\!\left((\theta,u,\beta)\in \mathcal F_n \mid Y_{1:n}\right)\stackrel{P_{\theta_0}^n}{\to} 1.
  $

  \item There exist constants $C_1,C_2>0$ and measurable sets $B_n \subseteq \Theta \times \mathcal U \times \mathbb R^p$ such that
  $
  \Pi_0(B_n)\ge e^{-C_1 n\varepsilon_n^2},
  $
  and for every $(\theta,u,\beta)\in B_n$,
  \(
  KL(p_{\theta_0},p_\theta)\le C_2\varepsilon_n^2,\quad
  V(p_{\theta_0},p_\theta)\le C_2\varepsilon_n^2,\quad
  G(\theta,u,\beta)\le C_2\varepsilon_n^2,
  \)
  where
  $
  V(p_{\theta_0},p_\theta)
  =
  \int \left[\log\!\left(\frac{p_{\theta_0}}{p_\theta}\right)\right]^2 p_{\theta_0}\,d\mu
  $ is the KL variance.

  \item There exist a deterministic sequence $\kappa_n>0$ and a constant $C_4>0$ such that, uniformly over $(\theta,u,\beta)\in \mathcal F_n$,
  $
  \ell_n(T(\beta))-\ell_n(\theta)
  \ge
  \kappa_n G(\theta,u,\beta)-C_4 n\varepsilon_n^2,
  $
  where $\ell_n(\theta)=\sum_{i=1}^n \log p_\theta(Y_i)$.

  \item  The constant $\alpha_n$ satisfies $
  \liminf_{n\to\infty} (\alpha_n+\kappa_n)\varepsilon_n^2 > 0
  $.
  \end{enumerate}
  Under (i)-(v), for every sufficiently large $M>0$,
  $
  \Pi_n\bigl(G(\theta,u,\beta)>M\varepsilon_n^2 \mid Y_{1:n}\bigr)\stackrel{P_{\theta_0}^n}{\to} 0.
  $
  \end{theorem}

Theorem \ref{thm:distance-relaxed} gives a useful guide for choosing the shrinkage parameter $\alpha_n$. The quantity
$\kappa_n$
measures how strongly the likelihood itself favors the exact projection over a nonzero-gap value, while $\alpha_n$ is the additional regularization supplied by the prior. Thus the relevant quantity is rather the combined strength
$\alpha_n+\kappa_n$. If the likelihood already contributes at the scale $\varepsilon_n^{-2}$, then fixed $\alpha_n\equiv \alpha$ is enough to preserve the exact projected-prior rate. In contrast, if the likelihood does not adequately penalize the duality gap, then $\alpha_n$ must increase with $n$ to compensate.

Next, we bound the distance between $T(\beta)$ and $\theta_0$ under the relaxed posterior. There is a key distinction between $\Pi^{\mathrm{ex},\beta}(\cdot\mid Y_{1:n})$ under the exact projected prior and $\Pi_n^{\mathrm{rel},\beta}(\cdot\mid Y_{1:n})$ under the relaxed prior: under the exact projected prior, the likelihood is evaluated at $\theta=T(\beta)$, whereas under the relaxed prior it is evaluated at a relaxation $\theta$ of $T(\beta)$.

\begin{theorem}\label{thm:beta-transfer}
Suppose there exists a measurable set $\mathcal B_n\subseteq \mathbb R^p$ and a deterministic sequence $\delta_n\to 0$ such that $\Pi_n^{\mathrm{ex},\beta}(\mathcal B_n\mid Y_{1:n})\stackrel{P_{\theta_0}^n}{\longrightarrow} 1$ and
\[\label{eq:beta-transfer-condition}
 P_{\theta_0}^n\left(\sup_{\beta\in\mathcal B_n}
\left|
\frac{m_n^{\mathrm{rel}}(\beta)}
{c_n \exp\{\ell_n(T(\beta))\}}-1
\right|
\le \delta_n\right)\to 1.
\]
for some deterministic $c_n>0$ and
$
m_n^{\mathrm{rel}}(\beta)
:=
\iint \exp\{\ell_n(\theta)\}\exp\{-\alpha_n G(\theta,u,\beta)\}\,d\theta\,du.
$
Then
$$
\left\|
\Pi_n^{\mathrm{rel},\beta}(\cdot\mid Y_{1:n})
-
\Pi_n^{\mathrm{ex},\beta}(\cdot\mid Y_{1:n})
\right\|_{\mathrm{TV}}
\stackrel{P_{\theta_0}^n}{\longrightarrow} 0.
$$
Consequently, if for sufficiently large $M>0$,
$
\Pi_n^{\mathrm{ex},\beta}\bigl(\rho(T(\beta),\theta_0)>M\varepsilon_n\mid Y_{1:n}\bigr)\stackrel{P_{\theta_0}^n}{\longrightarrow} 0,
$
then also
$
\Pi_n^{\mathrm{rel},\beta}\bigl(\rho(T(\beta),\theta_0)>M\varepsilon_n\mid Y_{1:n}\bigr)\stackrel{P_{\theta_0}^n}{\longrightarrow} 0.
$
\end{theorem}

The condition \eqref{eq:beta-transfer-condition} can hold under reasonable assumptions, such as appropriate local smoothness of the likelihood around the zero-gap point. In the interest of keeping the main exposition clear, we defer further theory development to the supplementary material.
With the above ingredients, we can now establish the contraction rate of the relaxed posterior.

\begin{theorem}\label{cor:relaxed-posterior-rate}
Assume conditions in Theorems \ref{thm:distance-relaxed} and \ref{thm:beta-transfer} hold, and further $\exists C>0:$
  $
  \rho(\theta,\theta_0)\le C\|\theta-T(\beta)\|_2+\rho(T(\beta),\theta_0).
  $
  Then, for sufficiently large $M>0$,
  $
  \Pi_n\bigl(\rho(\theta,\theta_0)>M\varepsilon_n\mid Y_{1:n}\bigr)\stackrel{P_{\theta_0}^n}{\longrightarrow} 0.
  $
  \end{theorem}
  An alternative to size-adaptive $\alpha_n$ is a two-stage approach, where a fixed value $\alpha_n = \alpha$ is used during posterior computation, and the projection $T(\beta)$ is applied in a post-processing step on the Markov chain samples after they have converged. Since this alternative is beyond the methodological scope of this article, we leave it for future work.

\section{Simulation Studies}
\subsection{Sparse regression}
We conduct experiments in a sparse linear regression setting with $n=200$ observations and $p=500$ predictors. The design matrix has $x_{ij}\sim \text{N}(0,1)$, and the true coefficient vector $\theta_0\in\mathbb{R}^p$ is sparse with $5$ nonzero entries drawn uniformly from $\{-4,-2,2,4\}$; the response $y_i=x_i^\top \theta_0+\varepsilon_i$ with $\varepsilon_i\sim \text{N}(0,1)$. We use the Inverse-Gamma$(2,1)$ prior for $\lambda$.

We compare three methods that admit data-augmented Gibbs samplers, the gap-shrinkage prior ($\alpha=1000$), the Bayesian lasso \citep{park2008bayesian}, and the Generalized Double Pareto (GDP) prior \citep{armagan2013generalized} in its hierarchical Laplace form; in addition, we also compare with the $\ell_1$-ball projected prior model. For these three methods the full conditionals are available in closed form, enabling exact Gibbs updates with no gradient evaluations. The $\ell_1$-ball projected prior \citep{xu2023l1ball} is fit with NUTS; its projection step (a sorting subroutine) must be executed at every  leapfrog evaluation, making it more expensive per iteration. Each experiment is repeated over 20 independent replications with 1000 warmup and 1000 sampling iterations.

Figure~\ref{fig:gap-shrinkage-prior}(b--c) shows violin plots of posterior samples of $\theta$ for 15 selected coefficients under the gap-shrinkage prior and the $\ell_1$-ball prior. The gap-shrinkage prior concentrates its mass near the ground truth and correctly identifies the sparse support, closely matching the exact $\ell_1$-ball projected posterior. The GDP posterior is similar and is not shown separately. The Bayesian lasso posterior shows inadequate shrinkage and biases for those nonzero coefficients. Figure~\ref{fig:gap-shrinkage-prior}(d) compares the effective sample size (ESS) per second across methods. The gap-shrinkage Gibbs sampler achieves the highest ESS per second (median $\approx 82$), followed by the Bayesian lasso Gibbs (median $\approx 62$) and GDP Gibbs (median $\approx 21$). The $\ell_1$-ball NUTS sampler is far slower (median $\approx 2$ ESS/sec) due to the overhead of repeated projection evaluations inside each leapfrog step. The autocorrelation function (ACF) boxplot in Figure~\ref{fig:gap-shrinkage-prior}(e) shows that the gap-shrinkage Gibbs sampler achieves rapid mixing.

\subsection{Low-rank and sparse matrix smoothing}
In the supplementary material, we show the simulations of a gap-shrinkage prior that bypasses expensive SVD operations.

\begin{figure}[H]
  \begin{subfigure}[t]{0.31\textwidth}
    \centering
    \includegraphics[width=1\textwidth]{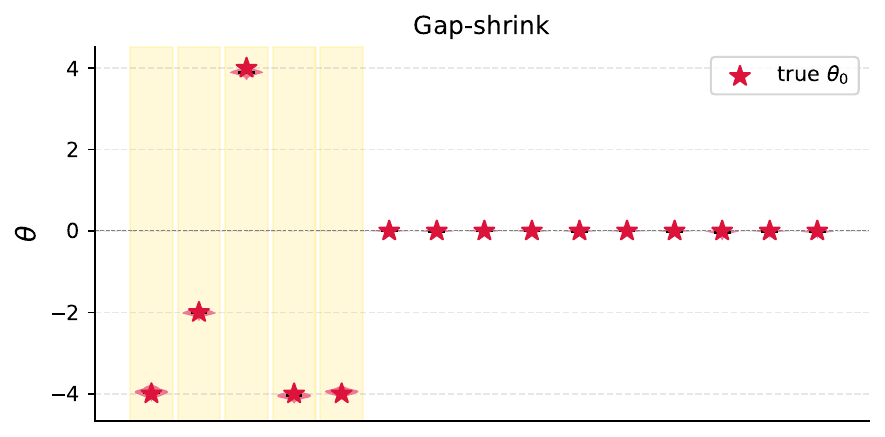}
    \caption{Posterior samples of $\theta$ under gap-shrinkage prior.}
  \end{subfigure}
  \hfill
  \begin{subfigure}[t]{0.31\textwidth}
    \centering
    \includegraphics[width=1\textwidth]{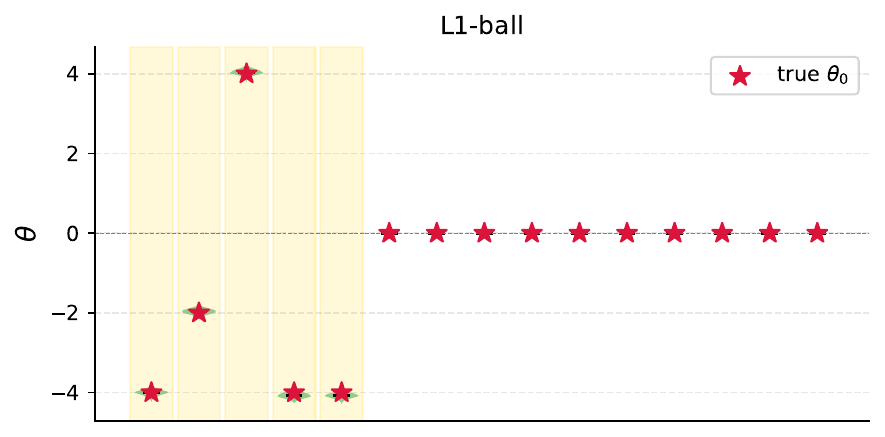}
    \caption{Posterior samples of $\theta$ under $\ell_1$-ball projected prior.}
  \end{subfigure}%
  \hfill
  \begin{subfigure}[t]{0.31\textwidth}
    \centering
    \includegraphics[width=1\textwidth]{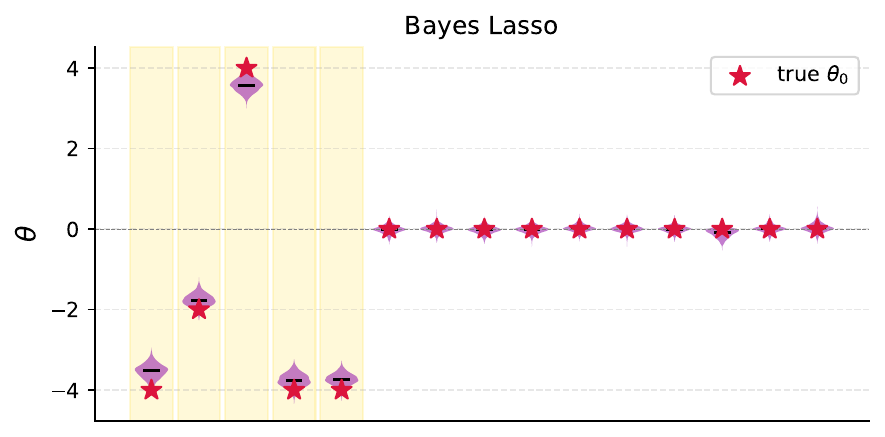}
    \caption{Posterior samples of $\theta$ under Bayesian lasso.}
  \end{subfigure}\\
  \begin{subfigure}[b]{0.5\textwidth}
    \centering
    \includegraphics[width=1\textwidth]{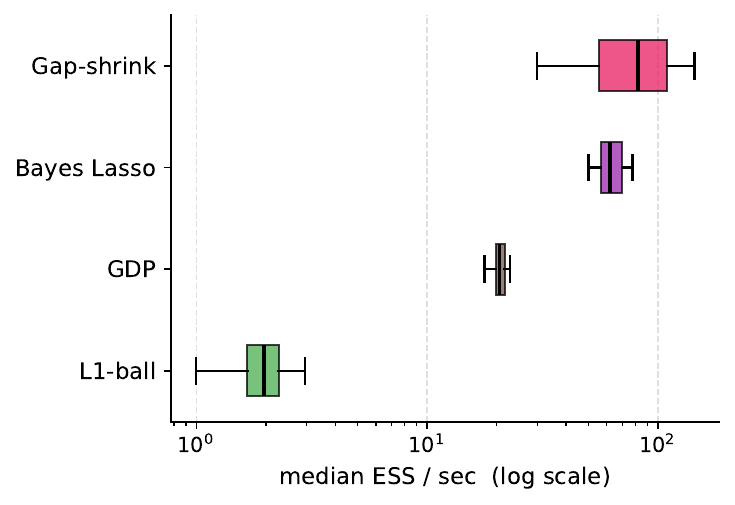}
    \caption{Median ESS per second.}
  \end{subfigure}%
  \hfill
  \begin{subfigure}[b]{0.48\textwidth}
    \centering
    \includegraphics[width=1\textwidth,height=4cm]{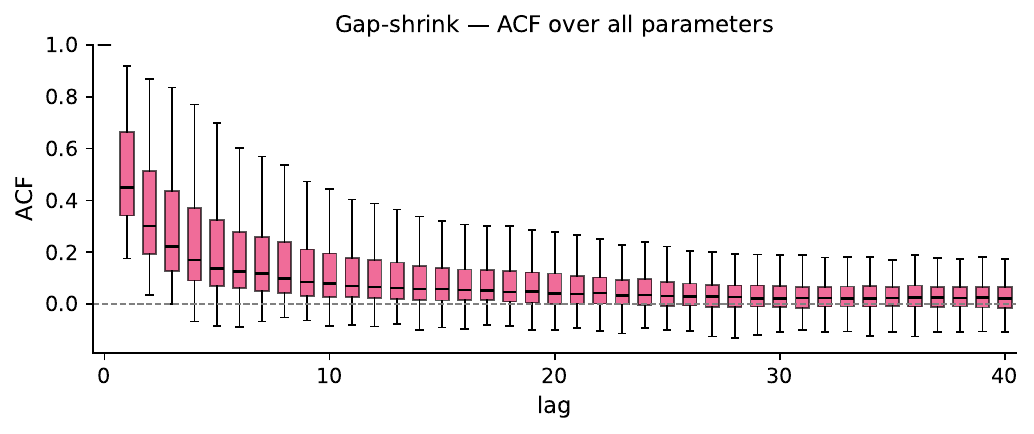}
    \caption{ACF boxplots (pooled over all $p$ parameters) under gap-shrinkage prior show rapid decay in autocorrelation.}
  \end{subfigure}
  \caption{Sparse regression results ($n=200$, $p=500$, 20 replications).}
  \label{fig:gap-shrinkage-prior}
\end{figure}
\vspace{-0.5cm}
\section{Data Application}
We now apply the gap-shrinkage model to household-level grocery purchase data from the 2020 Nielsen Consumer Panel.
We use the shopping records from $H=793$ households who made $n=803$ trips to grocery retailers, purchasing items from $m=26$ product subcategories under six departments: alcoholic beverages, dairy, deli, fresh produce, frozen foods, and packaged meat. Each response is a binary vector $y_i\in\{0,1\}^m$ indicating whether each of the $m=26$ subcategories was purchased on trip $i$. The data include several covariates: standardized total spending, total quantity, average paid price, quarters (Q1,Q2,Q3,Q4), and regions (Midwest, Northeast, South, West) derived from store ZIP code; in total, they correspond to $p=9$. For these multivariate binary responses, we use the probit model
\(
 y_{ij} \sim \text{Bernoulli}[\Phi(\mu_{ij})],\qquad
 \mu_i = \theta x_i + W\gamma_{h(i)},
\)
where $i$ is shopping trip, $j$ is product subcategory,  $h(i)$ is household for trip $i$, and $\mu_i \in \mathbb R^m$.
$\theta\in \mathbb{R}^{m\times p}$ is the category-by-covariate coefficient matrix, $\Phi$ is the standard Gaussian distribution function, $W\gamma_{h(i)}$ captures the household-specific random effect, $W\in\mathbb R^{m\times d}$ is a category loading matrix, and $\gamma_h\in [0,\infty)^d$ is a household-specific latent effect.

Since the coefficient matrix $\theta$ has $234$ elements, we impose a taxonomy-based smoothing prior across categories. The taxonomy is shown in the supplementary material. Let $G=(V,E)$ be a complete weighted graph on the $m$ subcategories, where $V=\{1,\ldots,m\}$. We fix the weight $\lambda_{j,j'}=1$ when the two subcategories belong to the same department and use a smaller weight $\lambda_{j,j'}=\omega_{\mathrm{cross}}\in(0,1)$ otherwise. Let $B\in\mathbb R^{|E|\times m}$ be a directed incidence matrix and let $\Lambda$ denote the corresponding diagonal matrix of edge weights. Then smoothing corresponds to the proximal loss
$
f(z;\beta)
=
\frac12\|\beta-z\|_F^2+\rho\|\Lambda B z\|_{1,1},
$
where $\|\cdot\|_{1,1}$ denotes the entrywise $\ell_1$ norm. For edge $e = (j, j')$, $B_{e,k} =
   +1$ if $k = j$, $B_{e,k} = -1$ if $k = j'$, and $B_{e,k} = 0$ otherwise. Each row of $Bz$ computes the difference $z_j - z_{j'}$ along one edge.
The corresponding gap-shrinkage prior is
\(
\Pi_0(\theta,v,\beta\mid \rho)
\propto
\exp\left[-\alpha\left\{\rho\|\Lambda B\theta\|_{1,1}-\langle \theta,B^\top \Lambda v\rangle\right\}\right]\pi^\beta_0(\beta),
\)
subject to $\theta=\beta-B^\top\Lambda v$ and $|v_{e,k}|\le \rho$. We set the prior for global smoothing strength $\rho\sim\text{Inverse-Gamma}(2,1)$, cross-department weight $\omega_{\mathrm{cross}}\sim\mathrm{Beta}(1,1)$, $d=5$ for the low-dimensional random effect structure, a Gaussian prior on $W$, and independent half-Gaussian priors on the coordinates of $\gamma_h$. Since $\theta$ is on the probit scale, heavy tails are not expected a priori, so we use a Gaussian prior kernel
$
\pi^\beta_0(\beta)\propto \exp\{-\|\beta\|_F^2/200\}
$
for computational convenience.
We ran $1000$ warmup iterations followed by $10000$ retained draws using NUTS. The fitted gap-shrinkage model yields a posterior mean of smoothing strength as $\rho\approx 0.0045$, while the cross-department weight is driven close to zero, with a posterior mean of $\omega_{\mathrm{cross}}\approx 2.8\times 10^{-4}$. Thus, the posterior favors smoothing within departments while applying almost no smoothing across departments. For the nine predictors, the posterior means of the coefficients are shown as boxplots in Figure~\ref{fig:nielsen-gap-violin}.

\begin{figure}[H]
  \begin{subfigure}[b]{0.45\textwidth}
    \centering
    \includegraphics[width=\textwidth,height=4cm]{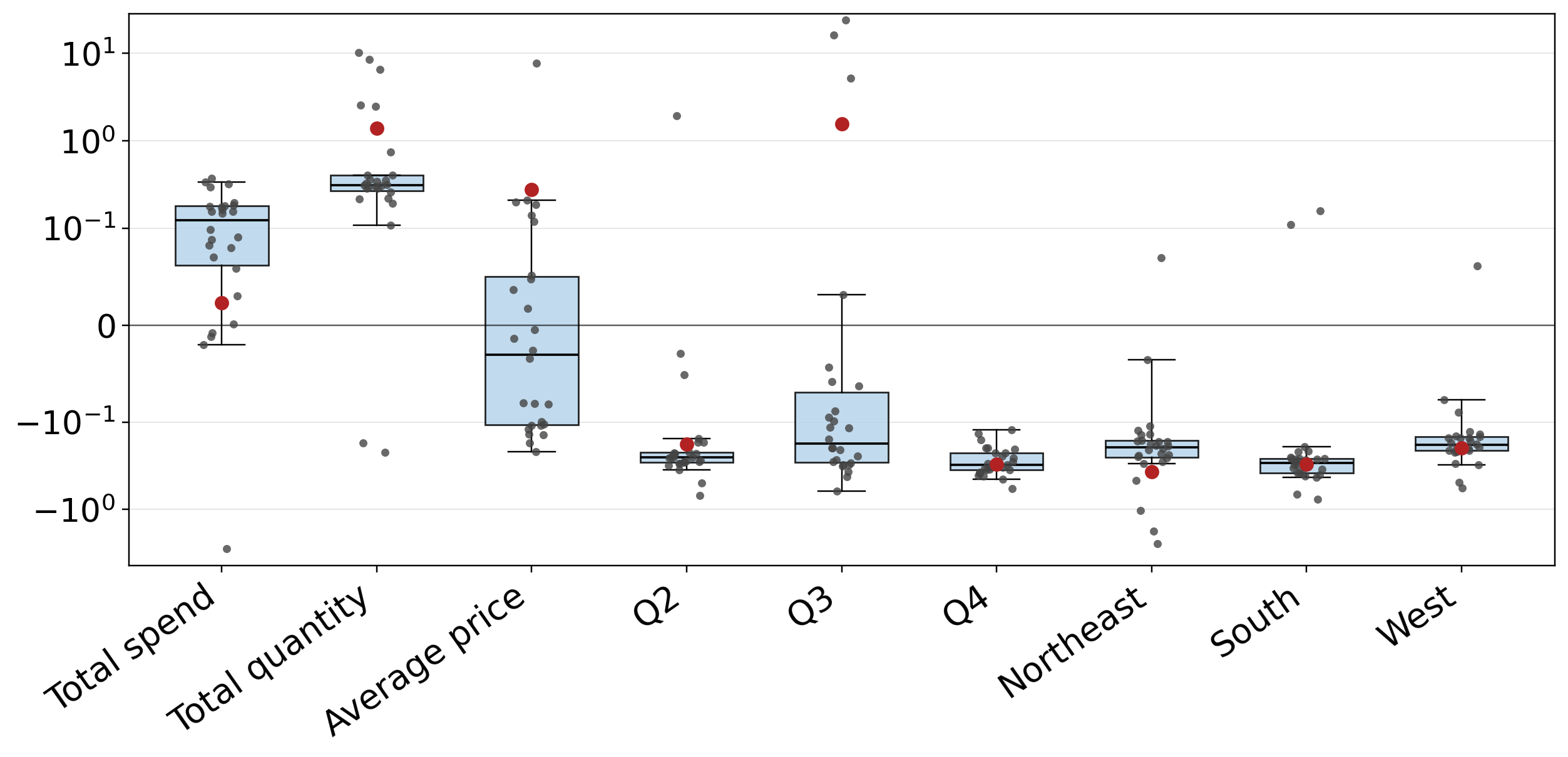}
    \caption{Distribution across the 26 subcategories of the posterior mean coefficient for each predictor under the gap-shrinkage model. The red dots mark the across-category posterior mean for each predictor.}
    \label{fig:nielsen-gap-violin}
  \end{subfigure}\qquad
  \begin{subfigure}[b]{0.45\textwidth}
    \centering
    \includegraphics[width=1\linewidth]{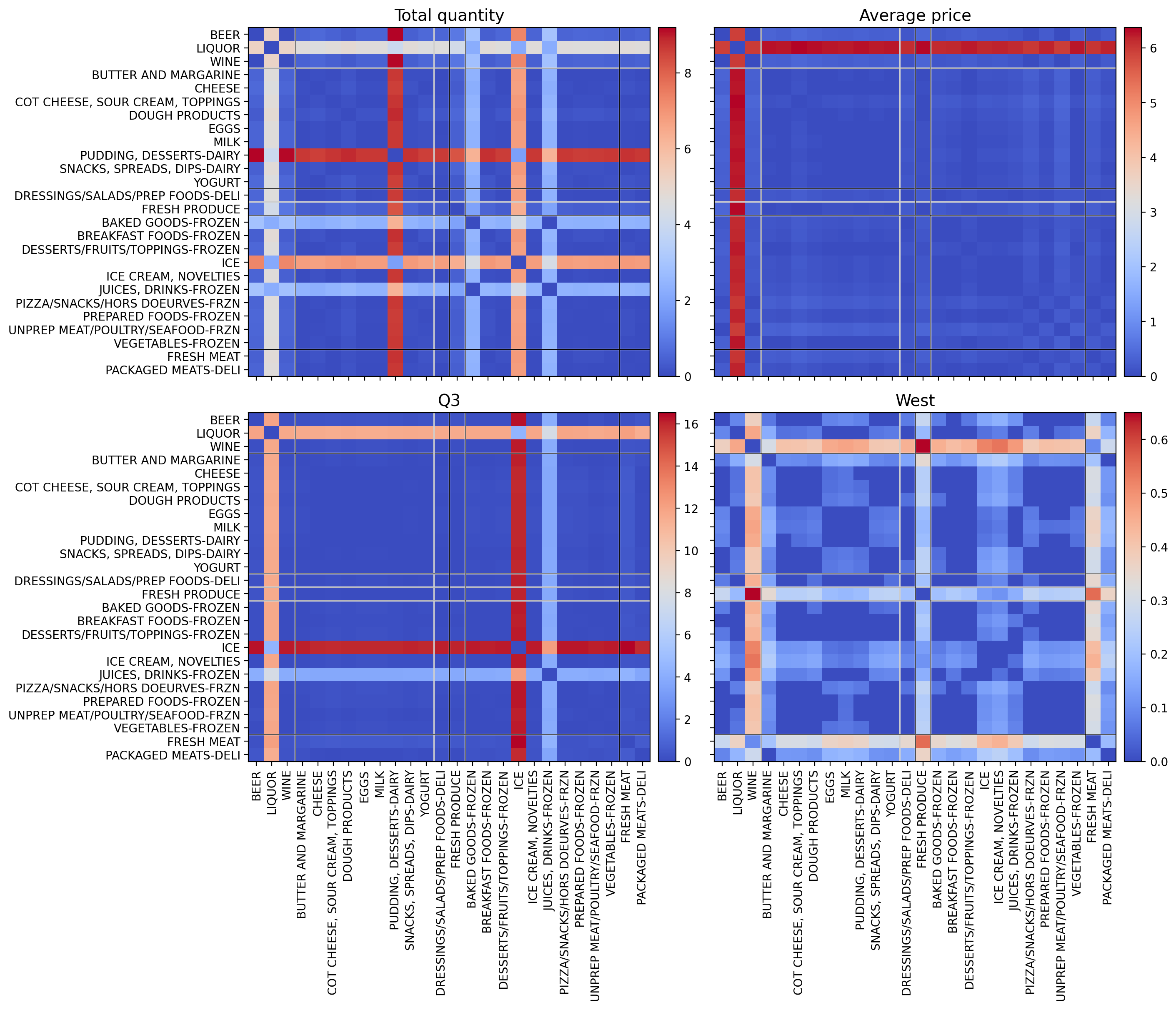}
    \caption{Four selected heatmaps of absolute pairwise posterior mean differences under the gap-shrinkage model, ordered by department.}
    \label{fig:nielsen-gap-heatmap_selected}
  \end{subfigure}
  \caption{(a) Posterior mean coefficients for each predictor by subcategory under the gap-shrinkage model; (b) Four selected heatmaps of absolute pairwise posterior mean differences under the gap-shrinkage model, ordered by department.}
\end{figure}
\vspace{-0.5cm}
Figure \ref{fig:nielsen-gap-heatmap_selected} shows four selected heatmaps of the absolute pairwise differences that are especially interpretable in this application. Each cell in a heatmap shows how differently two subcategories respond to a covariate. For total quantity, PUDDING, DESSERTS-DAIRY stands out because its posterior mean coefficient is strongly positive relative to nearly all other subcategories. For average price, LIQUOR is the dominant outlier, which is consistent with the fact that liquor purchases tend to be associated with higher average prices. For the third-quarter indicator (July--September), ICE shows an exceptionally large positive coefficient, which distinguishes it from the other seasonal indicators reported in the supplementary material, suggesting distinctly higher purchase probability in Q3. By contrast, the West region indicator is much more homogeneous overall; however, WINE remains noticeably more negative than FRESH PRODUCE and most other subcategories.

For comparison, we fit two alternate models, with the same likelihood and latent household term, but different priors on $\theta$: (i) a Bayesian fused-lasso prior with
$
\pi^\theta_0(\theta\mid \rho)\propto \exp\{-\rho\|\Lambda B\theta\|_{1,1}\},
$
and retaining the same prior for $\rho$ and $\omega_{\mathrm{cross}}$; (ii) an independent Gaussian prior with $\pi^\theta_0(\theta)\propto \exp\{-\|\theta\|_F^2/200\}$. Under the Bayesian fused-lasso prior, the model drives $\rho$ to an extremely large scale, with posterior mean about $9.94\times 10^6$, while pushing the cross-department weight to a posterior mean of about $10^{-8}$. As a result, the coefficient estimate under the Bayesian fused-lasso prior is much more homogeneous than the gap-shrinkage prior within department. The predictive ranking performance of the Bayesian fused-lasso model is also weaker. Using posterior mean probit scores and pooling all trip-category outcomes, the area under the ROC curve (AUC) is $0.867$ for the Bayesian fused-lasso model, compared with $0.897$ for the gap-shrinkage model. The independent Gaussian prior on $\theta$ gives AUC $0.887$ but produces much more dispersed category-level coefficients. We show more details in the supplementary material.
\vspace{-0.5cm}

\section{Discussion}
There are several future directions worth further exploration. First, we considered projection onto a convex set. It would be interesting to consider projection onto nonconvex sets as well, for example, those arising in principal-curve \citep{hastie1989principal} or predictively-optimal parameter space \citep{mclatchie2025predictively}, although in such settings identifiability and local geometry would require  careful treatment. Second, the connection between gap-shrinkage and global-local prior literature suggests a new view of Bayesian shrinkage: instead of shrinking coefficients toward zero, one can shrink them toward a structured set, in our setting, the image of a projection map. The idea could extend to any set satisfying an optimality condition.
\vspace{0.5cm}

\noindent\textbf{Declaration of the use of generative AI and AI-assisted technologies:}
During the preparation of this work the authors used Claude Code in order to assist with code development and figure plotting. After using this tool the authors reviewed and edited the content as necessary and take full responsibility for the content of the publication.

\appendix

\medskip
{\centering\huge\textbf{Supplementary Material}}

\renewcommand{\thesection}{\Alph{section}}

\section{Proofs}

\subsection{Tail bound for gap-shrinkage prior for $\ell_1$-norm proximal loss}
\begin{lemma}
  For the marginal prior $\pi_0(\theta | \lambda) = \prod_{j=1}^p \pi_0(\theta_j | \lambda)$, for $|\theta_j|$ sufficiently large, there exists a constant $C>0$ such that,
  $
    \pi_0(\theta_j | \lambda) \ge C {|\theta_j|^{-3}}.
  $
\end{lemma}
\begin{proof}
  We first note that $1+ (\theta_j+u_j)^2 \le 1+ (|\theta_j|+\lambda)^2$. Therefore,
  \(
& \int \exp\left\{- \alpha (\lambda-|u_j|) |\theta_j|  \right\}
1(|u_j|\le \lambda)\{1+(\theta_j+u_j)^2\}^{-1}\,\mathrm{d}u_j \\
& \ge
\frac{1}{1+(|\theta_j|+\lambda)^2}
\int_{-\lambda}^{\lambda}
\exp\left\{- \alpha (\lambda-|u_j|) |\theta_j| \right\}\,\mathrm{d}u_j \\
& =
\frac{2}{1+(|\theta_j|+\lambda)^2}
\int_0^\lambda
\exp\left\{- \alpha (\lambda-u_j) |\theta_j| \right\}\,\mathrm{d}u_j \\
& =
\frac{2}{1+(|\theta_j|+\lambda)^2}
\int_0^\lambda
e^{-\alpha t |\theta_j|}\,\mathrm{d}t \\
& =
\frac{2\left(1-e^{-\alpha\lambda|\theta_j|}\right)}
{\alpha |\theta_j|\{1+(|\theta_j|+\lambda)^2\}}.
\)
  Clearly, the right hand side is $\Omega(|\theta_j|^{-3})$ as $|\theta_j|\to \infty$.
\end{proof}

\subsection{Proof of Theorem~\ref{thm:bregman-shrinkage-effect}}
\begin{proof}
  Letting $\eta=\nabla\psi(\beta)$, the primal is
  \(
  \min_{z\in\mathbb R^p}\;
  \psi(z)-\psi(\beta)- (z-\beta)^\top \eta + I_C(z)
  =\min_z \Big\{\psi(z)- z^\top \eta + I_C(z)\Big\}
  +\beta^\top \eta-\psi(\beta).
  \)
  Letting $h(z)=\psi(z)- z^\top \eta$, we have
  $h^*(v)
  =
  \sup_z \{ v^\top z-\psi(z)+ z^\top \eta\}
  =
  \psi^*(v+\eta),
  $
  hence
  \(
  h(z)= \sup_{v} v^\top z - h^*(v)=\sup_{v} v^\top z - \psi^*(v+\eta).
  \)
  Ignoring constant $\beta^\top \eta-\psi(\beta)$, the effective primal is
  \(
  \inf_z \sup_{v} v^\top z - \psi^*(v+\eta) + I_C(z).
  \)
  Exchanging $\inf$ and $\sup$, we have the effective dual problem:
  \(
  & \sup_{v}  \inf_z  v^\top z - \psi^*(v+\eta) + I_C(z)=
  \sup_{v} \left\{ -\sup_z  (-v^\top z - I_C(z))\right\}  - \psi^*(v+\eta) \\
  & =\sup_{v} \left\{ -\sigma_C(-v)\right\}  - \psi^*(v+\eta).
  \)
  Replacing $v$ by $-u$ and adding the constant term, we obtain the dual function as stated in the theorem.

  Next, using the generalized Pythagorean inequality in Bregman projection, we have
  $$D_\psi(z,\beta)\ge D_\psi(z,\hat z)+D_\psi(\hat z,\beta),
  \qquad z\in C.$$
  Rearranging terms and applying the weak duality, we have the stated inequality.
  \end{proof}

  \subsection{Proof of Theorem~\ref{thm:consistency-relaxed}}
  \begin{proof}
    The proof is straightforward, and we present it for completeness.
    Let $\widetilde{\Pi}$ denote the prior on $(\theta,u,\beta)$ induced by the gap-shrinkage construction, and let $\Pi$ denote its marginal prior on $\theta$. By assumption, for every $\eta>0$,
    $$
    \widetilde{\Pi}\!\left(
    \{(\theta,u,\beta): G(\theta,u,\beta)=0,\ KL(p_{\theta_0},p_\theta)<\eta\}
    \right)>0.
    $$
    Marginalizing over $u$ and $\beta$, we obtain
      $
    \Pi\!\left(\{\theta: KL(p_{\theta_0},p_\theta)<\eta\}\right)>0
    \;\text{for every }\eta>0.
    $
    Hence the marginal prior $\Pi$ assigns positive mass to every KL neighborhood of $\theta_0$.
    Since the model is assumed to satisfy the conditions of Schwartz's consistency theorem, the KL-support property above implies posterior consistency under the prior $\Pi$. Therefore, for every $\epsilon>0$,
    $
    \Pi_n\!\left(\rho(\theta,\theta_0)>\epsilon \mid Y_{1:n}\right)\stackrel{P_{\theta_0}^n}{\to} 0.
    $
    \end{proof}

  \subsection{Proof of Theorem~\ref{thm:distance-relaxed}}
  \begin{proof}
  Write the relaxed posterior as
  \(
  \Pi_n(A\mid Y_{1:n})
  =
  \frac{N_n(A)}{D_n},
  \)
  where
  \(
  N_n(A)
  =
  \iiint_A \exp\{\ell_n(\theta)\} e^{-\alpha_n G(\theta,u,\beta)}\,\Pi_0(d\theta,du,d\beta),\quad
  D_n
  =
  N_n(\Theta\times\mathcal U\times\mathbb R^p).
  \)

  \noindent\textbf{1. Lower bound the denominator.}
  Observe that
$$
D_n
=
e^{\ell_n(\theta_0)}
\iiint e^{\ell_n(\theta)-\ell_n(\theta_0)} e^{-\alpha_n G(\theta,u,\beta)}\,\Pi_0(d\theta,du,d\beta).
$$
By assumption (iii), for $(\theta,u,\beta)\in B_n$ we have
$
G(\theta,u,\beta)\le C_2\varepsilon_n^2,
$
and therefore
$
D_n
\ge
e^{\ell_n(\theta_0)} e^{-\alpha_n C_2\varepsilon_n^2}
\iiint_{B_n} e^{\ell_n(\theta)-\ell_n(\theta_0)}\,\Pi_0(d\theta,du,d\beta).
$

  By the standard marginal likelihood lower bound over Kullback--Leibler neighborhoods; see, e.g., \citet{ghosal2000convergence}, assumption (iii) implies that there exists a constant $C>0$ such that
  $$
  P_{\theta_0}^n\!\left(
  \iiint_{B_n} e^{\ell_n(\theta)-\ell_n(\theta_0)}\,\Pi_0(d\theta,du,d\beta)
  \ge e^{-C n\varepsilon_n^2}
  \right)\to 1.
  $$
  Hence,
  $$P_{\theta_0}^n\!\left(
  D_n \ge e^{\ell_n(\theta_0)}
  \exp\{-C n\varepsilon_n^2-\alpha_n C_2\varepsilon_n^2\}
  \right)\to 1.
    $$

  \noindent\textbf{2. Upper bound the numerator.}
  Consider the numerator over the large-gap set
  \(
  A_{n,M}:=\{(\theta,u,\beta)\in \mathcal F_n: G(\theta,u,\beta)>M\varepsilon_n^2\}.
  \)
  For $(\theta,u,\beta)\in A_{n,M}$, assumption (iv) gives
  $
  \ell_n(\theta)
  \le
  \ell_n(T(\beta))-\kappa_n G(\theta,u,\beta)+C_4 n\varepsilon_n^2,
  $ hence
  \(
  e^{\ell_n(\theta)} e^{-\alpha_n G(\theta,u,\beta)}
  \le
  e^{\ell_n(T(\beta))}
  e^{-(\alpha_n+\kappa_n)G(\theta,u,\beta)+C_4 n\varepsilon_n^2}\le
  e^{C_4 n\varepsilon_n^2-(\alpha_n+\kappa_n)M\varepsilon_n^2}
  e^{\ell_n(T(\beta))}.
  \)

  Therefore,
  \(
  N_n(A_{n,M})
  \le
  e^{C_4 n\varepsilon_n^2-(\alpha_n+\kappa_n)M\varepsilon_n^2}
  \iiint_{A_{n,M}} e^{\ell_n(T(\beta))}\,\Pi_0(d\theta,du,d\beta).
  \)

  Now integrate out $(\theta,u)$ conditionally on $\beta$. Since the integrand depends on $(\theta,u)$ only through the indicator of $A_{n,M}$ and the prior,
  \(
  N_n(A_{n,M})
  \le
  e^{C_4 n\varepsilon_n^2-(\alpha_n+\kappa_n)M\varepsilon_n^2}
  \int e^{\ell_n(T(\beta))}\,m_n(\beta)\,\Pi_0^\beta(d\beta),
  \)
  where
  $
  m_n(\beta):=\Pi_0\bigl(\{(\theta,u):(\theta,u,\beta)\in A_{n,M}\}\mid \beta\bigr)\le 1.
  $
  Hence
  \(
  N_n(A_{n,M})
  \le
  e^{C_4 n\varepsilon_n^2-(\alpha_n+\kappa_n)M\varepsilon_n^2}
  \int e^{\ell_n(T(\beta))}\,\Pi_0^\beta(d\beta).
  \)

  Rewriting the last integral
  $$
  N_n(A_{n,M})
  \le
  e^{\ell_n(\theta_0)}
  e^{C_4 n\varepsilon_n^2-(\alpha_n+\kappa_n)M\varepsilon_n^2}
  \int e^{\ell_n(T(\beta))-\ell_n(\theta_0)}\,\Pi_0^\beta(d\beta).
  $$
Since
  $
  E_{\theta_0}\!\left[
  \int e^{\ell_n(T(\beta))-\ell_n(\theta_0)}\,\Pi_0^\beta(d\beta)
  \right]
  =
  \int E_{\theta_0}\!\left[
  \prod_{i=1}^n \frac{p_{T(\beta)}(Y_i)}{p_{\theta_0}(Y_i)}
  \right]\Pi_0^\beta(d\beta)
  =
  1.
  $
  Hence, by Markov's inequality, for every fixed $C'>0$,
  $
  P_{\theta_0}^n\!\left(
  \int e^{\ell_n(T(\beta))-\ell_n(\theta_0)}\,\Pi_0^\beta(d\beta)
  >
  e^{C' n\varepsilon_n^2}
  \right)
  \le
  e^{-C' n\varepsilon_n^2}.
  $
  Therefore, with high $P_{\theta_0}^n$-probability,
  $
  N_n(A_{n,M})
  \le
  e^{\ell_n(\theta_0)}
  \exp\!\left\{
  (C_4+C')n\varepsilon_n^2-(\alpha_n+\kappa_n)M\varepsilon_n^2
  \right\}.
  $

  \noindent\textbf{3. Combining the numerator and denominator bounds.}
We now have, with high $P_{\theta_0}^n$-probability,
  $$
  \Pi_n(A_{n,M}\mid Y_{1:n})
  \le
  \exp\!\left\{
  (C_4+C'+C)n\varepsilon_n^2
  +\alpha_n C_2\varepsilon_n^2
  -(\alpha_n+\kappa_n)M\varepsilon_n^2
  \right\}.
  $$
  Thus, if $M$ is sufficiently large, by condition (v)
  $
  \liminf_{n\to\infty} (\alpha_n+\kappa_n) > 0,
  $
  then
  $
  \Pi_n(A_{n,M}\mid Y_{1:n})\to 0.
  $

  Since
$
\{G(\theta,u,\beta)>M\varepsilon_n^2\}
\subseteq
A_{n,M}\cup \mathcal F_n^c,
$
we have
$
\Pi_n\bigl(G(\theta,u,\beta)>M\varepsilon_n^2 \mid Y_{1:n}\bigr)
\le
\Pi_n(A_{n,M}\mid Y_{1:n})
+
\Pi_n(\mathcal F_n^c\mid Y_{1:n}).
$
The first term tends to zero by the preceding bound, while the second tends to zero by assumption (ii). Hence
$
\Pi_n\bigl(G(\theta,u,\beta)>M\varepsilon_n^2 \mid Y_{1:n}\bigr)\to 0.
$
  \end{proof}

  \subsection{Proof of Theorem~\ref{thm:beta-transfer}}
\begin{proof}
Write
\(
\Pi_n^{\mathrm{rel},\beta}(d\beta\mid Y_{1:n})
=
\frac{m_n^{\mathrm{rel}}(\beta)\,\Pi_0^\beta(d\beta)}
{\int m_n^{\mathrm{rel}}(\beta)\,\Pi_0^\beta(d\beta)}
\quad\text{ and }\quad
\Pi_n^{\mathrm{ex},\beta}(d\beta\mid Y_{1:n})
=
\frac{e^{\ell_n(T(\beta))}\,\Pi_0^\beta(d\beta)}
{\int e^{\ell_n(T(\beta))}\,\Pi_0^\beta(d\beta)}.
\)
On the set $\mathcal B_n$, the assumption implies
\(
m_n^{\mathrm{rel}}(\beta)
=
c_n e^{\ell_n(T(\beta))}\{1+r_n(\beta)\},
\qquad
\sup_{\beta\in\mathcal B_n}|r_n(\beta)|\le \delta_n,
\)
with $P_{\theta_0}^n$-probability tending to one. Since $c_n$ does not depend on $\beta$, it cancels after normalization. Therefore, on $\mathcal B_n$, the two posterior densities differ only by $1+r_n(\beta)$, uniformly convergent to one. Therefore,
\(
& \left\|
\Pi_n^{\mathrm{rel},\beta}(\cdot\mid Y_{1:n})
-
\Pi_n^{\mathrm{ex},\beta}(\cdot\mid Y_{1:n})
\right\|_{\mathrm{TV}}\\
& \le
\left\|
\Pi_n^{\mathrm{rel},\beta}(\,\cdot\,\cap \mathcal B_n\mid Y_{1:n})
-
\Pi_n^{\mathrm{ex},\beta}(\,\cdot\,\cap \mathcal B_n\mid Y_{1:n})
\right\|_{\mathrm{TV}}
+
\Pi_n^{\mathrm{rel},\beta}(\mathcal B_n^c\mid Y_{1:n})
+
\Pi_n^{\mathrm{ex},\beta}(\mathcal B_n^c\mid Y_{1:n})\\
& \stackrel{P_{\theta_0}^n}{\longrightarrow} 0.
\)

For the final claim, let
$
A_{n,M}:=\{\beta:\rho(T(\beta),\theta_0)>M\varepsilon_n\}.
$
Then
\(
\left|
\Pi_n^{\mathrm{rel},\beta}(A_{n,M}\mid Y_{1:n})
-
\Pi_n^{\mathrm{ex},\beta}(A_{n,M}\mid Y_{1:n})
\right|
\le
\left\|
\Pi_n^{\mathrm{rel},\beta}(\cdot\mid Y_{1:n})
-
\Pi_n^{\mathrm{ex},\beta}(\cdot\mid Y_{1:n})
\right\|_{\mathrm{TV}},
\)
which converges to zero in $P_{\theta_0}^n$-probability. Hence contraction rate of $T(\beta)$ under the exact projected posterior transfers to the relaxed posterior.
\end{proof}

\subsection{Sufficient conditions for the $\beta$-marginal convergence}

We now proceed to give a set of sufficient conditions for the $\beta$-marginal convergence in Theorem~\ref{thm:beta-transfer}. The content is that the relaxed inner integral over $(\theta,u)$ admits a local Laplace approximation around the zero-gap point $(T(\beta),u^*(\beta))$.

  \begin{theorem}\label{prop:laplace-comparison}
  Let
  \(
  J_n(\beta)
  =
  \iint \exp\{\ell_n(\theta)-\ell_n(T(\beta))\}\exp\{-\alpha_n G(\theta,u,\beta)\}\,d\theta\,du.
  \)
  Suppose there exists a measurable set $\mathcal B_n\subseteq \mathbb R^p$ such that
  $
  \Pi_n^{\mathrm{ex},\beta}(\mathcal B_n\mid Y_{1:n})\stackrel{P_{\theta_0}^n}{\longrightarrow} 1,
  $
   and the following conditions hold uniformly over $\beta\in\mathcal B_n$.

  \begin{enumerate}
\renewcommand{\labelenumi}{(\roman{enumi})}
  \item There is a unique zero-gap point
  $
  z^*(\beta)=(T(\beta),u^*(\beta))
  $
  with
  $
  G(z^*(\beta),\beta)=0.
  $

  \item There exists a positive definite matrix $H_\beta$ such that, for every fixed $M>0$,
  \(
  \sup_{\|h\|\le M/\sqrt{\alpha_n}}
  \left|
  G(z^*(\beta)+h,\beta)-\frac12 h^\top H_\beta h
  \right|
  =
  o(\alpha_n^{-1}).
  \)

  \item The Hessians are uniformly nondegenerate and asymptotically constant:
  $
  0<\lambda_{\min}\le \lambda_{\min}(H_\beta)\le \lambda_{\max}(H_\beta)\le \lambda_{\max}<\infty,
  $
  and there exists a deterministic sequence $d_n>0$ such that
  $
  \det(H_\beta)^{-1/2}=d_n\{1+o(1)\}.
  $

  \item For every fixed $M>0$,
  \(
  \sup_{\beta\in\mathcal B_n}
  \sup_{\|\theta-T(\beta)\|\le M/\sqrt{\alpha_n}}
  |\ell_n(\theta)-\ell_n(T(\beta))|
  \stackrel{P_{\theta_0}^n}{\longrightarrow} 0.
  \)
  \item There exists $c_0>0$ such that, for all sufficiently large $n$,
  $$
  G(z^*(\beta)+h,\beta)\ge c_0\|h\|^2
  $$
  whenever $\|h\|$ is sufficiently small, uniformly in $\beta\in\mathcal B_n$. Moreover, for every $\eta>0$,
  $$
  \lim_{M\to\infty}\,\limsup_{n\to\infty}
  P_{\theta_0}^n\!\left(
  \sup_{\beta\in\mathcal B_n}
  \alpha_n^{m/2}J_{n,2}(\beta;M)
  \ge \eta
  \right)=0,
  $$
  with $J_{n,2}(\beta;M)$ a tail integral defined as in the proof.
  \end{enumerate}

  Then there exists a deterministic sequence
  $
  c_n=(2\pi)^{m/2}\alpha_n^{-m/2}d_n
  $
  such that
  \(
  \sup_{\beta\in\mathcal B_n}
  \left|
  \frac{J_n(\beta)}{c_n}-1
  \right|
  \stackrel{P_{\theta_0}^n}{\longrightarrow} 0 \; \Leftrightarrow \; \sup_{\beta\in\mathcal B_n}
  \left|
  \frac{m_n^{\mathrm{rel}}(\beta)}
  {c_n e^{\ell_n(T(\beta))}}-1
  \right|
  \stackrel{P_{\theta_0}^n}{\longrightarrow} 0.
  \)
  \end{theorem}

      \begin{proof}
        Fix $\beta\in\mathcal B_n$, and write
        $
        z=(\theta,u),
        \;
        z^*=z^*(\beta)=(T(\beta),u^*(\beta)).
        $
        Then
        $
        J_n(\beta)
        =
        \int \exp\{\ell_n(\theta)-\ell_n(T(\beta))\}e^{-\alpha_n G(z,\beta)}\,dz.
        $

        Make the change of variable
        $
        h=z-z^*.
        $
        Then
        \(
        & J_n(\beta)
        =
        \int \exp\{\ell_n(T(\beta)+h_\theta)-\ell_n(T(\beta))\}
        e^{-\alpha_n G(z^*+h,\beta)}\,dh,\\
        & = \underbrace{\int_{\|h\|\le M/\sqrt{\alpha_n}}
        \exp\{\ell_n(T(\beta)+h_\theta)-\ell_n(T(\beta))\}
        e^{-\alpha_n G(z^*+h,\beta)}\,dh}_{J_{n,1}(\beta;M)}\\
        & \qquad +     \underbrace{\int_{\|h\|> M/\sqrt{\alpha_n}}
        \exp\{\ell_n(T(\beta)+h_\theta)-\ell_n(T(\beta))\}
        e^{-\alpha_n G(z^*+h,\beta)}\,dh}_{J_{n,2}(\beta;M)}.
        \)
        where $h_\theta$ denotes the $\theta$-component of $h$.

        Fix $\eta>0$. By assumption (v), there exists $M<\infty$ sufficiently large such that
        $$
        \sup_{\beta\in\mathcal B_n} J_{n,2}(\beta;M)
        \le \eta\,\alpha_n^{-m/2}
        $$
        with $P_{\theta_0}^n$-probability tending to one. For this fixed $M$, it remains to approximate $J_{n,1}(\beta;M)$.

        On the set $\|h\|\le M/\sqrt{\alpha_n}$, assumption (iv) implies
        $$
        \sup_{\beta\in\mathcal B_n}
        \sup_{\|h\|\le M/\sqrt{\alpha_n}}
        \left|
        \exp\{\ell_n(T(\beta)+h_\theta)-\ell_n(T(\beta))\}-1
        \right|
        \to 0
        $$
        in $P_{\theta_0}^n$-probability. Also, by assumption (ii),
        $$
        \sup_{\beta\in\mathcal B_n}
        \sup_{\|h\|\le M/\sqrt{\alpha_n}}
        \left|
        \alpha_n G(z^*(\beta)+h,\beta)
        -\frac12 (\sqrt{\alpha_n}h)^\top H_\beta (\sqrt{\alpha_n}h)
        \right|
        \to 0
        $$
        as $n\to\infty$. Therefore,
        \(
        \sup_{\beta\in\mathcal B_n}
        \left|
        J_{n,1}(\beta;M)
        -
        \int_{\|h\|\le M/\sqrt{\alpha_n}}
        \exp\!\left\{-\frac{\alpha_n}{2}h^\top H_\beta h\right\}\,dh
        \right|
        =
        o_{P_{\theta_0}^n}(\alpha_n^{-m/2}).
        \)

        Now set
        $
        t=\sqrt{\alpha_n}\,h.
        $
        Then
        $$
        \sup_{\beta\in\mathcal B_n}
        \left|
        J_{n,1}(\beta;M)
        -
        \alpha_n^{-m/2}
        \int_{\|t\|\le M}
        \exp\!\left\{-\frac12 t^\top H_\beta t\right\}\,dt
        \right|
        =
        o_{P_{\theta_0}^n}(\alpha_n^{-m/2}).
        $$
        uniformly over $\beta\in\mathcal B_n$.

        Combining this with the tail bound gives
        $$
        \sup_{\beta\in\mathcal B_n}
        \left|
        J_n(\beta)
        -
        \alpha_n^{-m/2}
        \int_{\|t\|\le M}
        \exp\!\left\{-\frac12 t^\top H_\beta t\right\}\,dt
        \right|
        \le
        \eta\,\alpha_n^{-m/2}
        +
        o_{P_{\theta_0}^n}(\alpha_n^{-m/2}).
        $$
        Since $\eta>0$ was arbitrary, letting $M\to\infty$ in the Gaussian integral and using assumption (iii), we conclude that
        $$
        J_n(\beta)
        =
        \alpha_n^{-m/2}(2\pi)^{m/2}\det(H_\beta)^{-1/2}
        \{1+o_{P_{\theta_0}^n}(1)\}
        $$
        uniformly over $\beta\in\mathcal B_n$.

        Finally, since
        $
        \det(H_\beta)^{-1/2}=d_n\{1+o(1)\}
        $
        uniformly over $\beta\in\mathcal B_n$, it follows that
        $$
        J_n(\beta)
        =
        (2\pi)^{m/2}\alpha_n^{-m/2}d_n\{1+o_{P_{\theta_0}^n}(1)\}
        $$
        uniformly over $\beta\in\mathcal B_n$.
        This proves the claim.      \end{proof}

        \begin{remark}
          In the above result, the regime $\alpha_n\to\infty$ is used to induce localization of the relaxed inner integral on the scale $\alpha_n^{-1/2}$ around the zero-gap point $(T(\beta),u^*(\beta))$.  In principle, the approximation
          $
          \sup_{\beta\in\mathcal B_n}
          \left|
          \frac{m_n^{\mathrm{rel}}(\beta)}
          {c_n e^{\ell_n(T(\beta))}}-1
          \right|
          \stackrel{P_{\theta_0}^n}{\longrightarrow} 0
          $
          may still hold without requiring $\alpha_n\to\infty$, but we conjecture that it may hold under different conditions ensuring that the relaxed inner integral is already concentrated near $(T(\beta),u^*(\beta))$ at a fixed scale and that the likelihood is sufficiently flat on that scale. We leave this as a potential future work.
          \end{remark}

          \subsection{Proof of Theorem~\ref{cor:relaxed-posterior-rate}}
          \begin{proof}
            By Theorem~\ref{thm:distance-relaxed}, the posterior duality gap is of order $\varepsilon_n^2$, and therefore Theorem~\ref{thm:shrinkage-effect} implies that
            $$
            \|\theta-T(\beta)\|_2=O_{P_{\theta_0}^n}(\varepsilon_n)
            $$
            under the relaxed posterior. By Theorem~\ref{thm:beta-transfer}, contraction of $T(\beta)$ under the exact projected posterior transfers to the relaxed posterior, so
            $$
            \Pi_n^{\mathrm{rel},\beta}\bigl(\rho(T(\beta),\theta_0)>M\varepsilon_n\mid Y_{1:n}\bigr)\stackrel{P_{\theta_0}^n}{\longrightarrow} 0
            $$
            for every sufficiently large $M>0$. The conclusion then follows from
            $$
            \rho(\theta,\theta_0)\le C\|\theta-T(\beta)\|_2+\rho(T(\beta),\theta_0).
            $$
            \end{proof}

            \section{Simulation on low-rank and sparse matrix smoothing}
We next consider a low-rank and sparse matrix smoothing problem, where the goal is to estimate a matrix $\theta\in\mathbb{R}^{p_1\times p_2}$ that is both low-rank and elementwise sparse, as in robust PCA. Projection onto this structure can be formulated via the proximal mapping with $g(\theta)=\lambda_1\|\theta\|_*+\lambda_2\|\theta\|_1$, where $\|\cdot\|_1$ is the elementwise $\ell_1$ norm and $\|\cdot\|_*$ is the nuclear norm. However, solving this projection is computationally intensive, requiring repeated SVDs and iterative algorithms like ADMM, which motivates our use of relaxed projected priors.

Combining the results developed from Examples 4 and 5, we have the variational gap
\(
    \tilde G(\theta=AB^\top,V_1,V_2;\beta)
    = \frac{\lambda_1}{2}\bigl(\|A\|_F^2+\|B\|_F^2\bigr)
      + \lambda_2\|AB^\top\|_1
      - \langle V_1+V_2,\,AB^\top\rangle,
  \)
  with both $V_1,V_2\in\mathbb{R}^{p_1\times p_2}$ and $\|V_1\|_{\mathrm{op}}\le\lambda_1$ and $|(V_2)_{ij}|\le\lambda_2$ for all $i,j$; $A\in \mathbb{R}^{p_1\times r}$ and $B\in \mathbb{R}^{p_2\times r}$ for some $r\ge 1$. Under sufficiently large $\lambda_1$, we know the exact projection $\theta$ has a low rank, hence allowing us to use $r$ smaller than $p_1$ and $p_2$.

To simulate data, we use $p_1=50$, $p_2=40$, and first construct a ground-truth matrix $\theta_0\in\mathbb{R}^{p_1\times p_2}$ of rank $3$ and elementwise sparsity $96.2\%$: three rank-one blocks of amplitudes $10$, $7$, and $4$ are placed at non-overlapping $5\times 5$ submatrices, and all other entries are zero. The data are drawn i.i.d. via
\(
  Y_s = \theta_0 + \varepsilon_s, \qquad \varepsilon_s\sim\mathrm{N}(0,\sigma^2 I_{p_1}\otimes I_{p_2}),
\)
for $s=1,\ldots,S$ with $S=100$, $\sigma=0.3$. We use matrix Gaussian likelihood consistent with the data simulation, and assign inverse-gamma prior for $\sigma^2$, and relaxed projection prior for the matrix $\theta$ with dimension $r=5$, and elements $\beta_{ij}\sim \mathrm{N}(0,100)$. For the operator norm constraint on $V_1$, we enforce it by simply reparameterizing $\lambda_1=\|V_1\|_F\ge \|V_1\|_{\mathrm{op}}$.

We run $3000$ warmup and $3000$ sampling iterations.
Figure~\ref{fig:nuclear-gibbs} shows the posterior mean of the estimated $\theta$, and the posterior distribution of the singular values $\rho_k(\theta)$ of $\theta$. The posterior singular values $\rho_k(\theta)$ concentrate tightly around the true values for $k=1,2,3$, while the posterior for $k\ge 4$ collapses to near zero, reflecting automatic rank selection. The noise variance $\sigma^2$ is recovered accurately, with posterior mean $0.092$. The results show that the gap-shrinkage prior handles simultaneous constraints and correctly identifies matrix rank and sparsity.

In terms of the computational costs, the Gibbs sampler for the gap-shrinkage prior took $29$ seconds. For comparison, we also ran the model with the exact projected prior, using ADMM to solve the projection step. This approach required $4$ hours to complete $6{,}000$ iterations, which is substantially slower than the gap-shrinkage Gibbs sampler.

\begin{figure}[H]
  \centering
  \begin{subfigure}[b]{0.45\textwidth}
    \centering
    \includegraphics[width=1\linewidth, height=5cm, trim=0 0 0 30, clip]{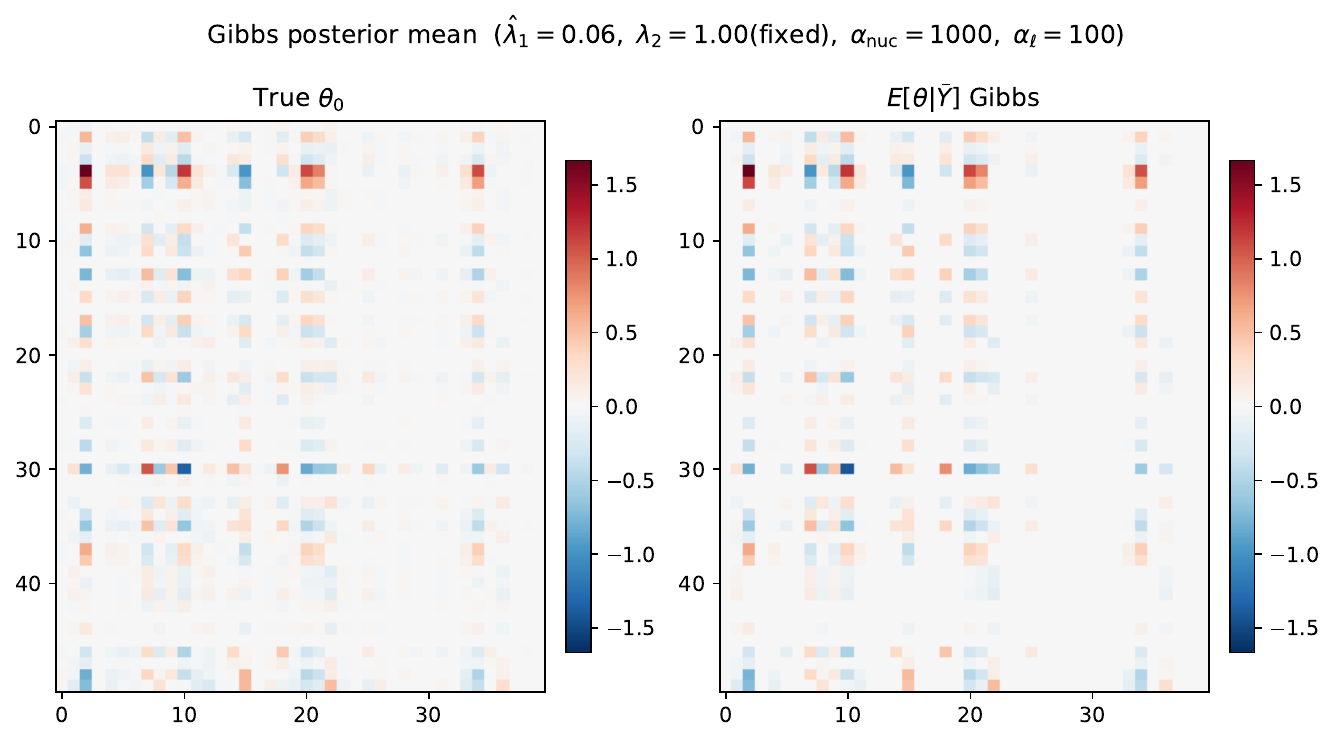}
    \caption{Ground truth $\theta_0$, and posterior mean of $\theta$ under gap-shrinkage prior.}
  \end{subfigure}
  \begin{subfigure}[b]{0.45\textwidth}
    \centering
    \includegraphics[width=1\linewidth, height=5cm, trim=0 0 0 20, clip]{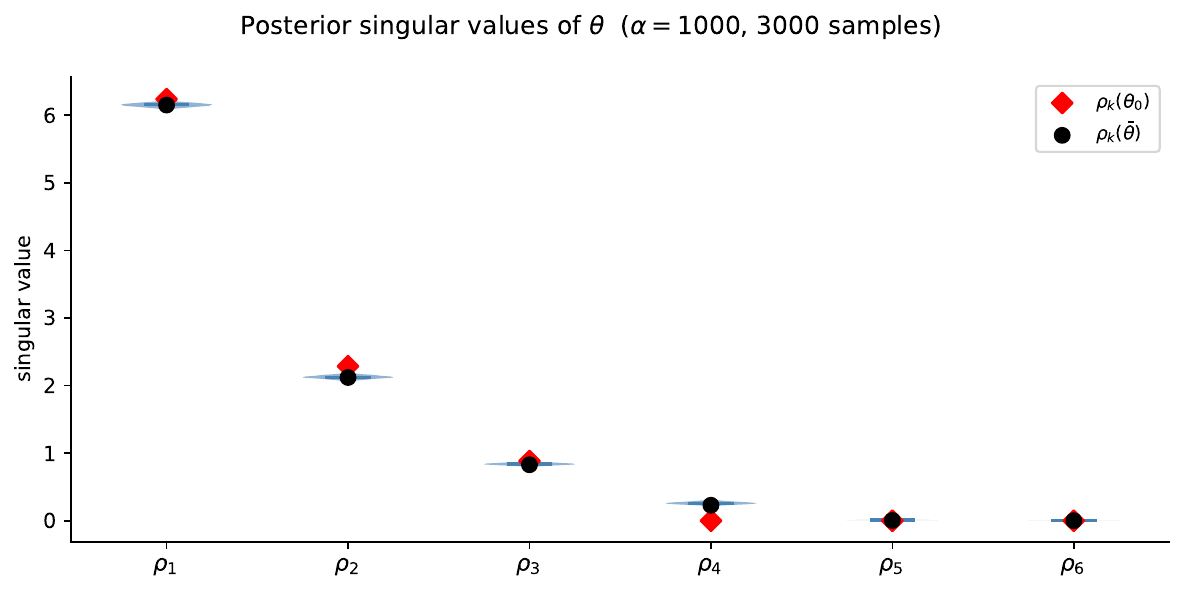}
    \caption{Posterior distribution of singular values $\rho_k(\theta)$ for $k=1,\ldots,6$.}
  \end{subfigure}
  \caption{Results for low-rank and sparse matrix smoothing.}
  \label{fig:nuclear-gibbs}
\end{figure}

\section{Additional details for the data application}
We performed a 5-fold cross-validation using the same three models. The pooled AUCs are $0.832$ for gap-shrinkage, $0.829$ for the Bayesian fused-lasso prior, and $0.821$ for the independent Gaussian prior. The foldwise distributions in Figure \ref{fig:nielsen-gap-roc-compare}(c) show that out-of-sample differences are modest, but the gap-shrinkage prior remains slightly better overall while avoiding the over-smoothing behavior of the Bayesian fused-lasso fit and high dispersion in coefficients produced by the independent Gaussian prior with little pooling.

\begin{figure}[H]
\centering
\begin{subfigure}[b]{0.45\textwidth}
    \centering
    \includegraphics[width=1\textwidth]{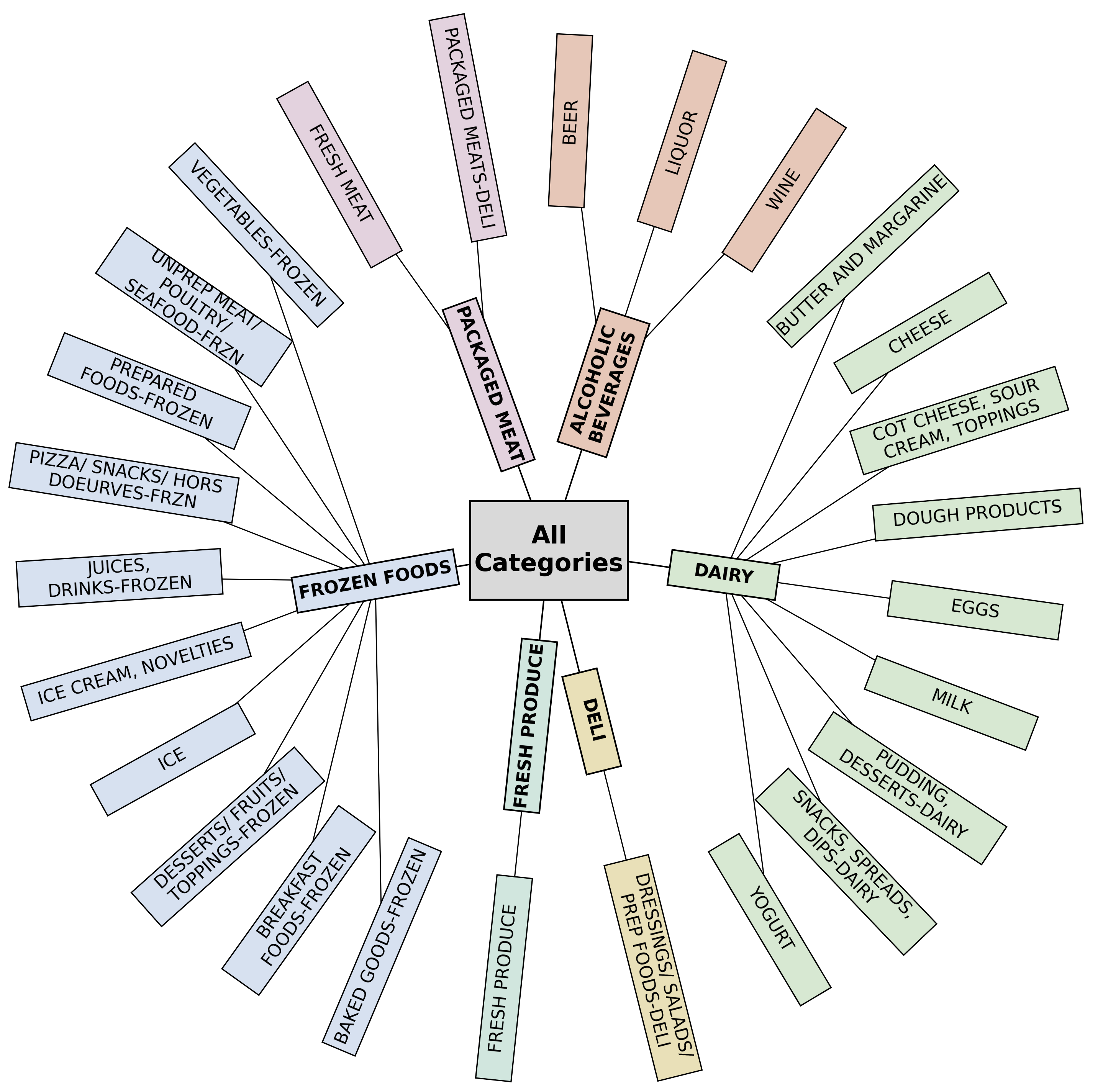}
    \caption{Taxonomy of the 26 subcategories in grocery.}
    \label{fig:nielsen-taxonomy}
  \end{subfigure}
\begin{subfigure}[b]{0.45\textwidth}
  \includegraphics[width=1\textwidth]{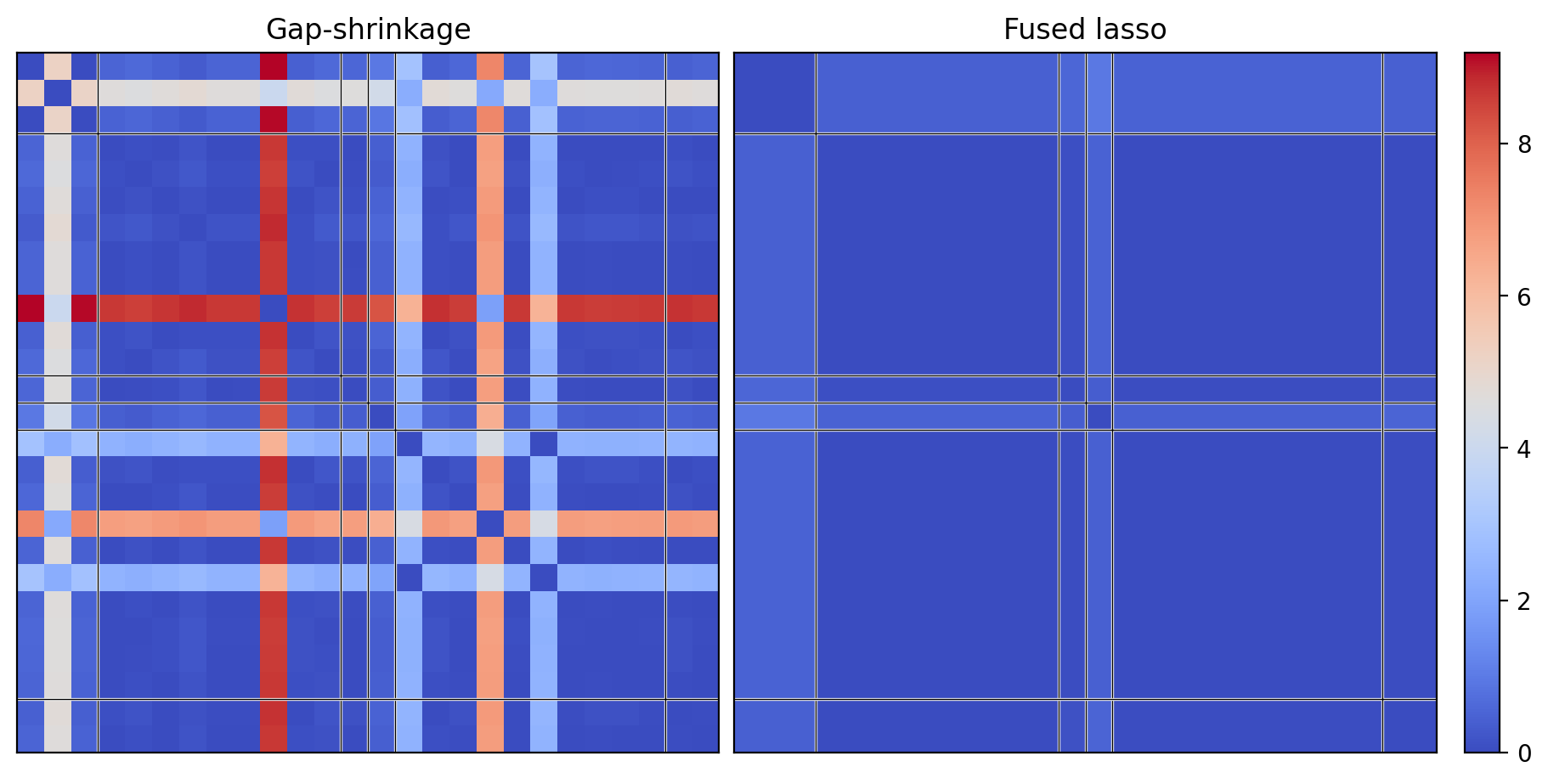}
  \caption{Absolute pairwise posterior mean differences for the total quantity coefficient.}
  \label{fig:nielsen-gap-vs-fused}
\end{subfigure}\\
\begin{subfigure}[b]{0.4\textwidth}
  \includegraphics[width=\textwidth,height=3cm]{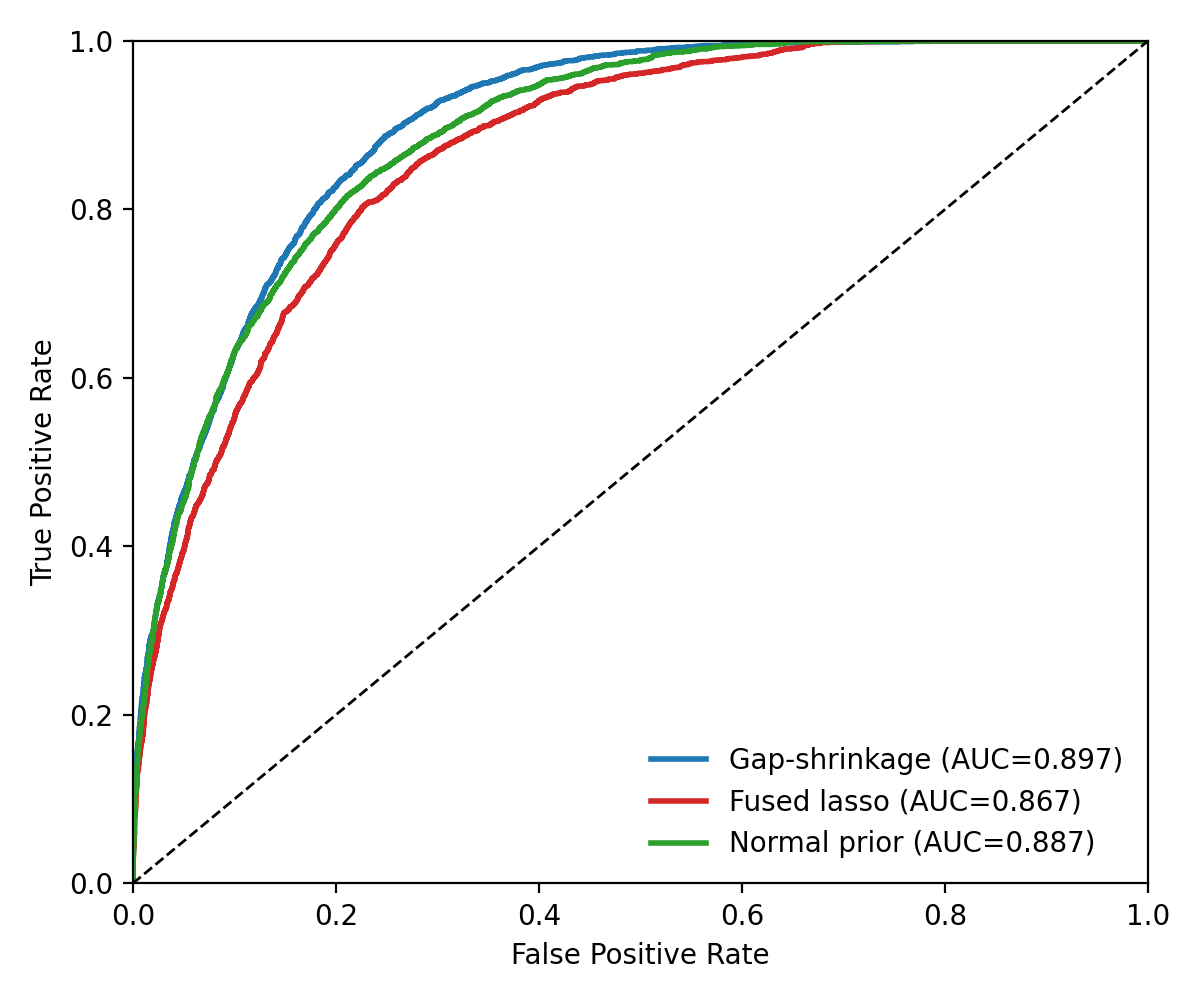}
  \caption{ROC curves based on posterior mean probit scores.}
  \label{fig:nielsen-roc}
\end{subfigure}
\begin{subfigure}[b]{0.4\textwidth}
  \centering
  \includegraphics[width=\textwidth,height=3cm]{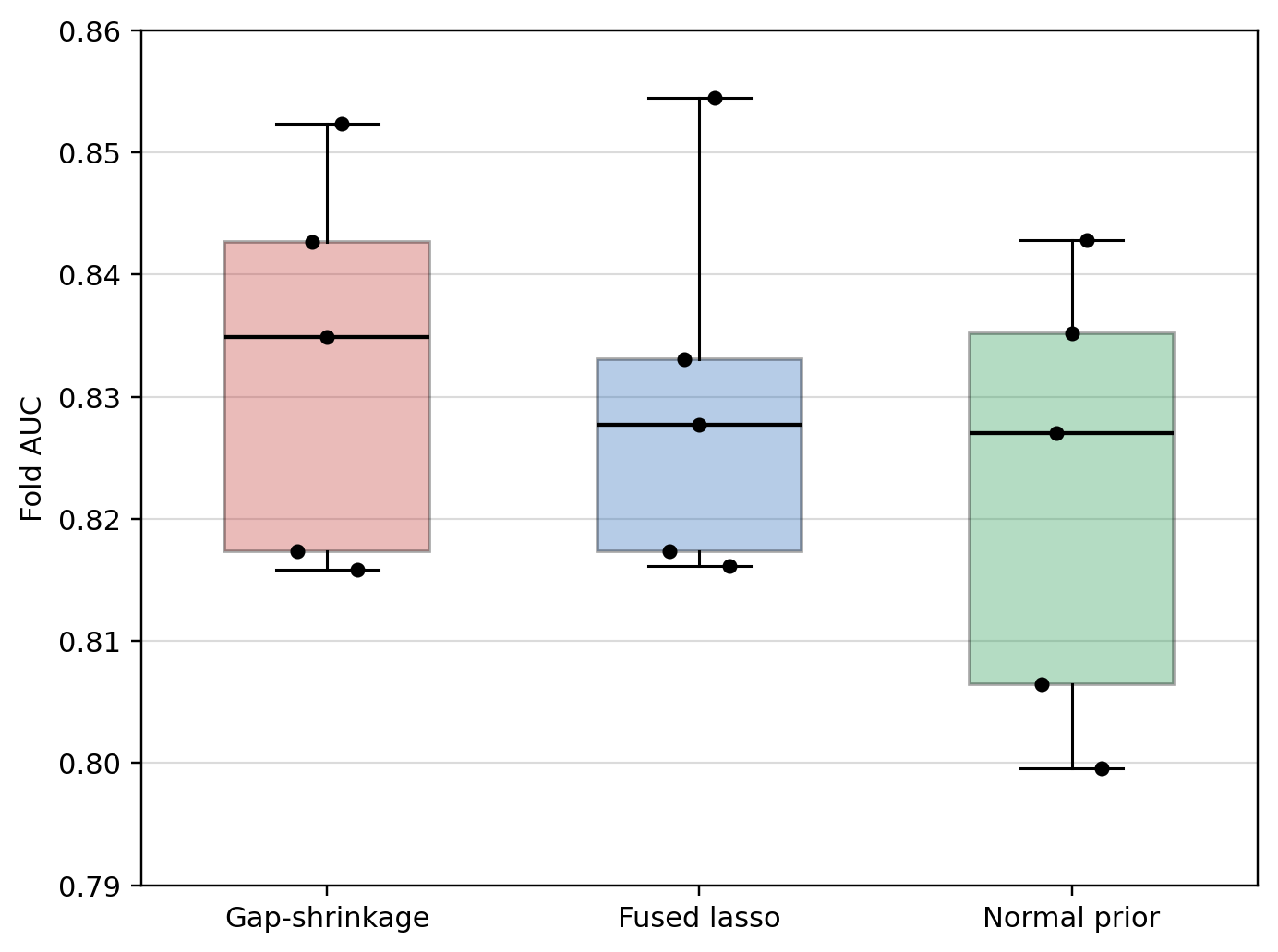}
  \caption{Boxplots of foldwise AUC values from the five-fold cross-validated comparison.}
\end{subfigure}
\caption{Comparison of the gap-shrinkage model with the Bayesian fused-lasso model and the independent Gaussian prior model on $\theta$.}
\label{fig:nielsen-gap-roc-compare}
\end{figure}

\begin{figure}[H]
\begin{subfigure}[t]{0.45\textwidth}
\centering
\includegraphics[width=\linewidth]{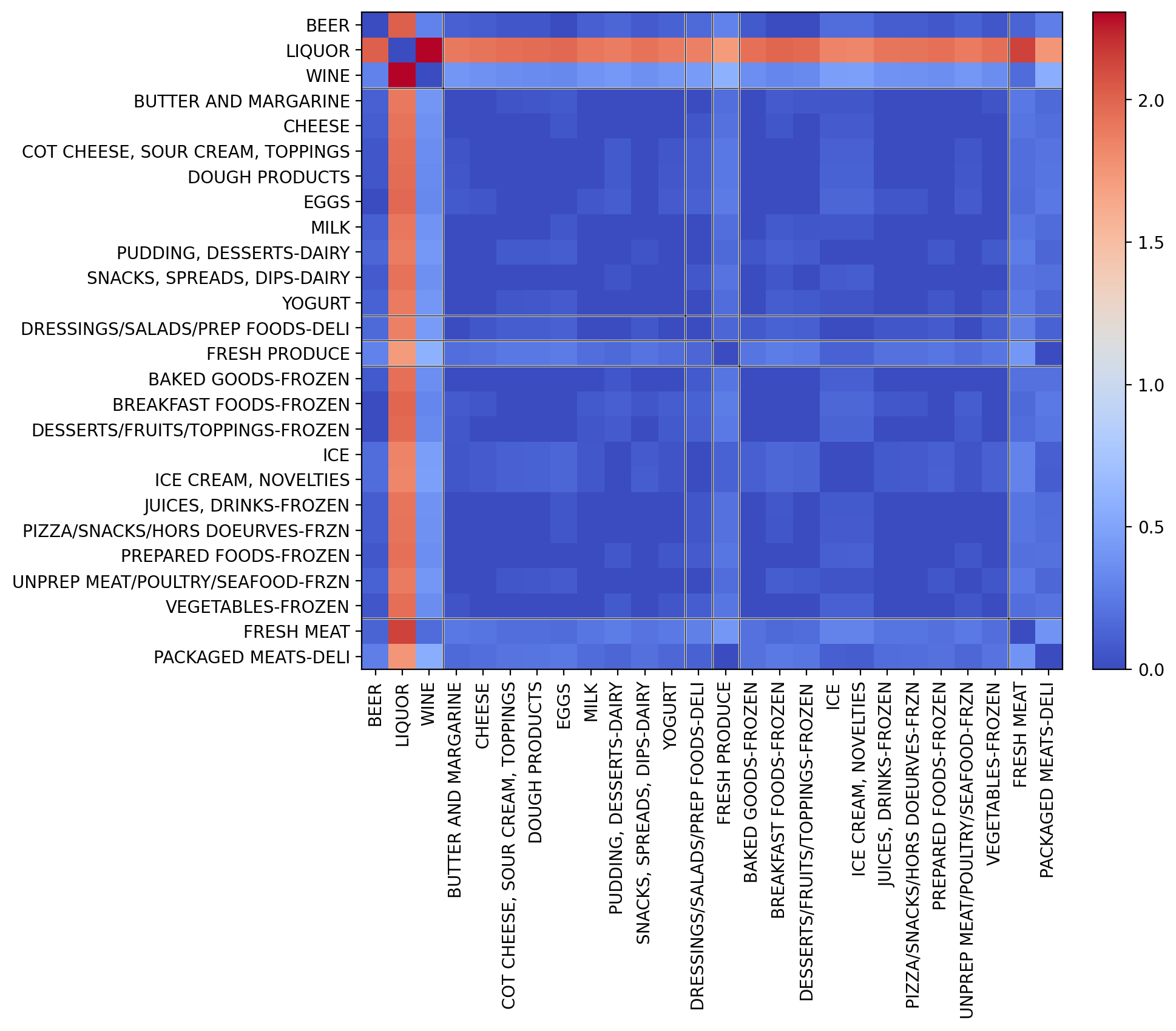}
\caption{$Q_2$}
\end{subfigure}
\begin{subfigure}[t]{0.45\textwidth}
\centering
\includegraphics[width=\linewidth]{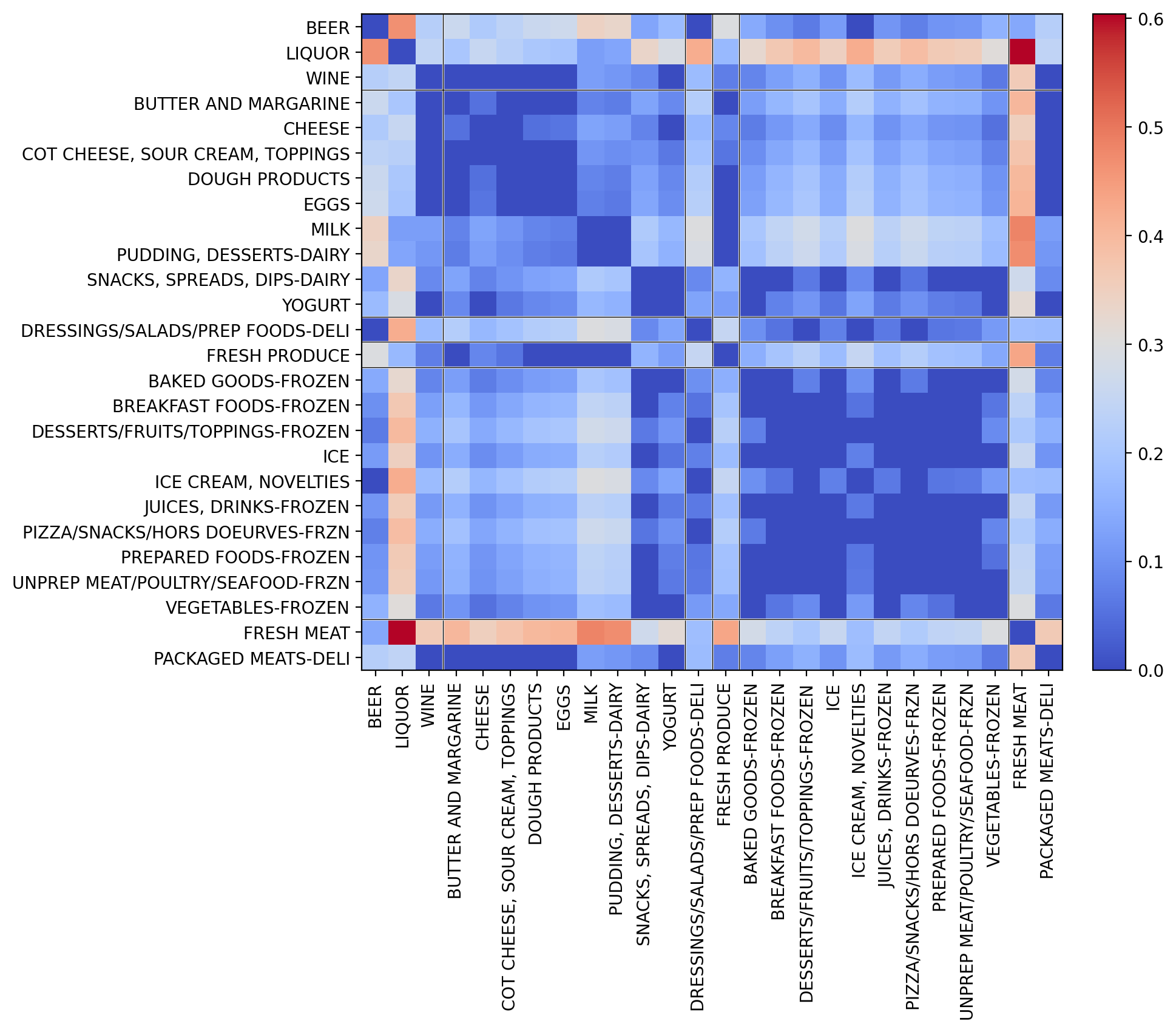}
\caption{$Q_4$}
\end{subfigure}
\begin{subfigure}[t]{0.45\textwidth}
\centering
\includegraphics[width=\linewidth]{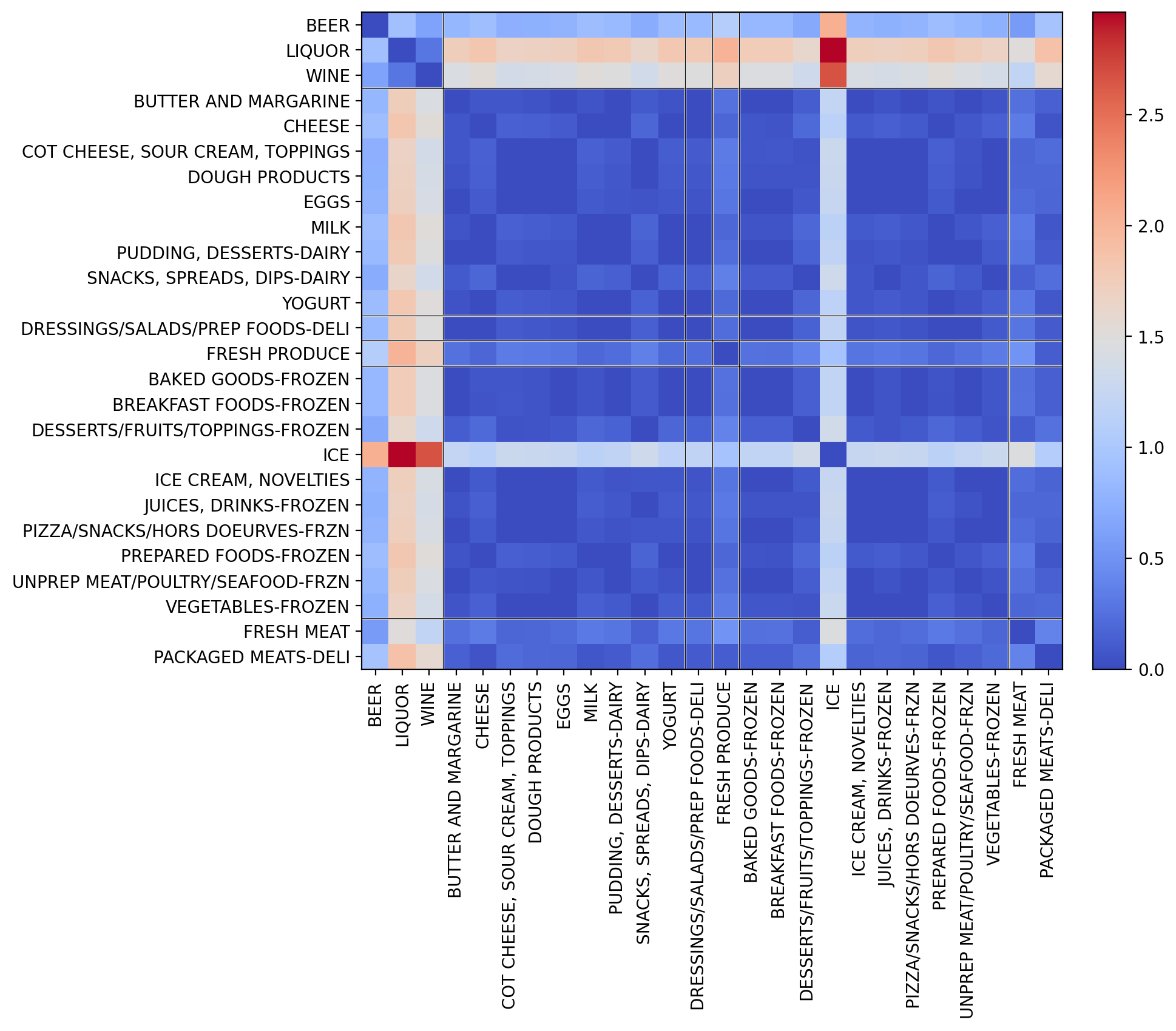}
\caption{Northeast}
\end{subfigure}
\begin{subfigure}[t]{0.45\textwidth}
\centering
\includegraphics[width=\linewidth]{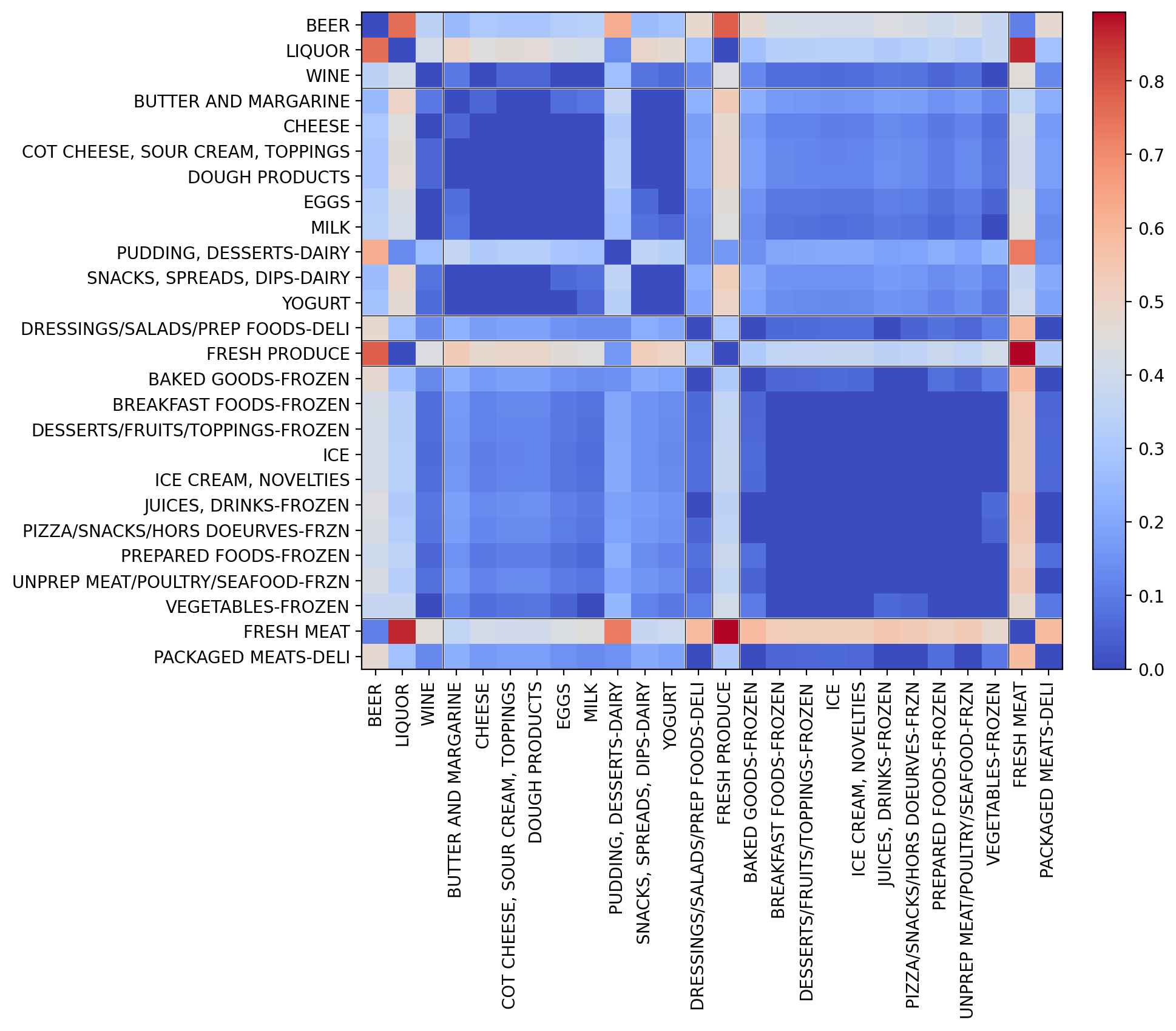}
\caption{South}
\end{subfigure}
\begin{subfigure}[t]{0.45\textwidth}
  \centering
  \includegraphics[width=\linewidth]{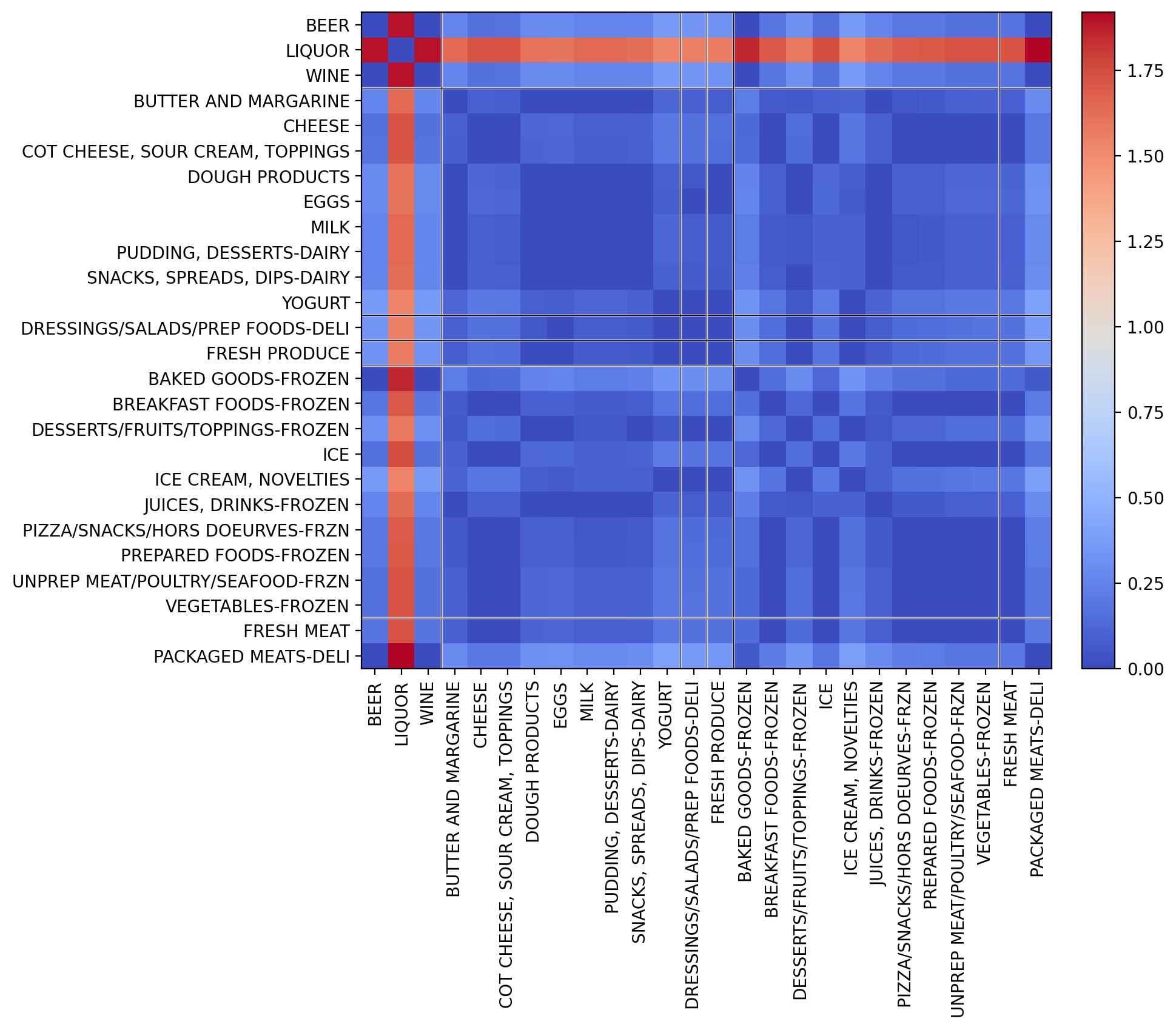}
  \caption{Total spend}
  \end{subfigure}
\caption{Absolute pairwise posterior mean differences under the gap-shrinkage model, ordered by department.}
\label{fig:nielsen-gap-heatmap}
\end{figure}

\bibliographystyle{plainnat}
\bibliography{ref_fixed}

\end{document}